\documentclass[11pt,a4paper]{amsart}
\usepackage{amsfonts}
\usepackage{latexsym,amssymb,amsfonts,graphicx}
\usepackage{amsmath,dsfont}
\usepackage[linewidth=1pt]{mdframed}
\usepackage{mathrsfs}
\usepackage[normalem]{ulem}
\usepackage{epsfig}
\usepackage{tabularx}
\usepackage{verbatim}
\usepackage{enumitem}
\usepackage{framed}
\usepackage{float}
\usepackage{algorithm2e}
\usepackage{tikz}
\usetikzlibrary{arrows.meta,positioning}
\usepackage[a4paper,left=25mm,right=25mm,top=20mm,bottom=20mm,headheight=15mm,headsep=10mm]{geometry}

\usepackage{url}
\usepackage{bm}
\usepackage{color}
\usepackage{natbib}
\setcitestyle{authoryear,round}

\usepackage[unicode=true,colorlinks,citecolor=blue,linkcolor=blue,urlcolor=blue]{hyperref}

\providecommand{\U}[1]{\protect\rule{.1in}{.1in}}

\newtheorem{thm}{Theorem}[section]

\newtheorem{definition}{Definition}[section]

\newtheorem{lem}[thm]{Lemma}

\newtheorem{rem}{Remark}[section]

\usepackage[bf,SL,BF]{subfigure}

\numberwithin{equation}{section}
\allowdisplaybreaks[4]

\newcommand{\revadd}[1]{\textcolor{black}{#1}}

\definecolor{linkcolor}{rgb}{0,0,1}
\hypersetup{colorlinks=true,citecolor=linkcolor,linkcolor=linkcolor,urlcolor=linkcolor}

\title[OPTIMAL CREDIT PORTFOLIO WITH DEFAULT CONTAGION]{OPTIMAL CREDIT PORTFOLIO AND CONSUMPTION\\ WITH REGIME SWITCHING AND DEFAULT CONTAGION}
\author[Fei Sun]{Fei Sun$^{1}$}
\author[Wenyuan Wang]{Wenyuan Wang$^{2,4}$}
\author[Kaixin Yan]{Kaixin Yan$^{3,4,*}$}
\thanks{$^{1}$School of Mathematics and Computational Science, Wuyi University, Jiangmen, 529020, China. $^{2}$School of Mathematics and Statistics, Fujian Normal University, Fuzhou, Fujian, 350117, China. $^{3}$School of Mathematics and Statistics, Xi'an Jiaotong University, Xi'an, Shaanxi, 710049, China. $^{4}$School of Mathematical Sciences, Xiamen University, Xiamen, Fujian, 361005, China. Fei Sun: \texttt{fsun.sci@outlook.com}; Wenyuan Wang: \texttt{wwywang@xmu.edu.cn}; Kaixin Yan: \texttt{kaixinyan@stu.xmu.edu.cn}. $^{*}$Corresponding author.}
\thanks{Fei Sun is supported by the National Natural Science Foundation of China (12401620), the Special Foundation in Key Fields for Universities of Guangdong Province (2023ZDZX4060), and the Education Science Planning Project of Guangdong Province (2024GXJK269).}
\thanks{Wenyuan Wang is supported by the National Natural Science Foundation of China (12571508, 12171405, 11661074), the Natural Science Foundation of Fujian Province (2024J01480), and the Program for New Century Excellent Talents in Fujian Province University (20720220044).}
\date{}

\begin{document}

\maketitle

\begin{abstract}
	We study an optimal portfolio and consumption problem in a regime-switching multi-name
	credit market with default contagion. Default events not only generate direct portfolio losses but also
	alter the default intensities of surviving securities. Under a Cobb-Douglas utility,
	the homogeneity property reduces the associated Hamilton-Jacobi-Bellman (HJB) equation to a recursive system of ordinary differential equations indexed by the
	default states. Solving this system backward from the all-default state, we establish
	existence and uniqueness of a positive classical solution, characterize the
	optimal feedback controls, and provide a verification theorem.  Finally, numerical experiments present sensitivity analyses and comparisons across contagion settings, regimes, utilities, as well as before and after default.
\end{abstract}


\section{Introduction}\label{sec:introduction}

		{\color{red}}

	The continuous-time portfolio-consumption problem originates in the seminal
		work of \citet{Merton69,Merton71}, who derive optimal rules for allocating wealth between
		consumption and risky investment under a diffusion model. Subsequent studies have extended this paradigm in several important directions. \citet{KaratzasLehoczkySethiShreve86} obtain explicit
		consumption-investment policies for multiple risky assets and general utility functions.
		\citet{CoxHuang89} employ a martingale approach to characterize the optimal consumption and portfolio
		policies when the asset prices follow diffusion processes.
		\citet{DavisNorman90}
		examine an infinite-horizon consumption-investment problem with proportional transaction
		costs in a market consisting of a bank account and a risky stock. 
		In an incomplete-market setting, \citet{Cuoco97} uses martingale methods to establish the existence of the optimal consumption-investment policies in the presence of stochastic income and portfolio constraints. For further studies of optimal consumption and investment, see, for example, \citet{FlemingZariphopoulou91,Zariphopoulou94,SchroderSkiadas99,Liu04,KraftSeiferlingSeifried17}.

The classical formulation assumes that the investment environment remains stable over time.
In real financial markets, however, expected returns, volatilities, and interest rates can vary significantly in response to changing macroeconomic conditions. Finite-state regime-switching models offer a tractable setting for capturing such persistent shifts. Within this framework, \citet{ZhouYin03} examine a continuous-time mean-variance portfolio selection problem in which the interest rate, appreciation rate, and volatility of multiple stocks are modulated by a finite-state Markov chain.
There is also a growing body of literature that investigates the coordinated optimization of consumption and investment. For instance, \citet{SotomayorCadenillas09} derive explicit optimal consumption and investment policies under hyperbolic absolute risk aversion (HARA)  utility in a regime-switching market. In a finite-horizon setting, \citet{Liu14} studies optimal consumption and investment in a regime-switching market, with an objective function that combines utility from both intermediate consumption and terminal wealth. For further contributions to optimal consumption and investment under
regime switching, see, for example, \citet{PirvuZhang14}, \citet{Wu16},
\citet{WeiShenZhao20}, \citet{ZhangLiangYuen21}, and
\citet{HuShiXu25}.
Taken together, these studies suggest that the optimal consumption and investment strategy is contingent upon the prevailing regime.

There is also a strand of literature that explores the optimal investment and consumption problem where investors are allowed to invest in defaultable bonds with/without default contagion. The credit risk inherent in such bonds introduces an additional source of state dependence. 
Specifically, once a name defaults, it enters an absorbing default state and is no longer tradable; through contagion, the default event may also alter the default intensities of the surviving names. On the single-name front,
\citet{Linetsky06} proposes an intensity-based bankruptcy model in which the stock price jumps to zero at bankruptcy, and derives
closed-form prices for corporate bonds and stock options. 
\citet{BieleckiJang06} obtain closed-form optimal strategies for an  investor who allocates wealth among a defaultable bond, a stock, and a savings account.
\citet{CapponiFigueroaLopez14} extend this single-default setting by allowing the market coefficients, default intensity, and loss given default to depend on a finite-state regime-switching process. \citet{BoLiWangYang13} study optimal investment and consumption problem in an infinite-horizon market with a defaultable bond and derive explicit optimal policies under HARA utility of consumption. For multiple defaults, \citet{JiaoKharroubiPham13} develop a BSDE-decomposition approach that reduces the optimal investment problem under contagion risk to a backward recursive system of quadratic BSDEs. 
Building on these contagion models, \citet{Bo19} incorporate regime switching into an insurer's joint credit-investment and risk-control problem; whereas \citet{BoLiaoYu19} replace the expected utility of terminal wealth with a risk-sensitive objective and analyze a finite-horizon portfolio problem in a similar credit market. 

However, none of the aforementioned studies address the coordinated optimization of consumption and investment in a market characterized by both regime switching and multiple defaults with contagion. To fill this gap, we investigate a finite-horizon coordinated optimization problem in such a market, where the objective function incorporates utility from intermediate consumption, intermediate wealth, and terminal wealth.
The investor must balance intermediate consumption against wealth retention for future consumption or credit investment. We capture this trade-off through a Cobb-Douglas running utility over intermediate consumption expenditure and wealth, paired with a power utility over terminal wealth. For this problem, we prove the existence and uniqueness of a positive classical solution to the corresponding recursive HJB system, obtain closed-form feedback strategy for optimal consumption and credit portfolio, and confirm its optimality through a verification theorem.

The main contributions of this paper are twofold. First, we are the first to simultaneously incorporate intermediate consumption, intermediate wealth, and terminal wealth into the objective functional of our portfolio-consumption problem (see \eqref{objective.cobb}). In particular, neither intermediate consumption nor intermediate wealth has been considered in any of the related multi-name credit-investment models discussed in the aforementioned literature. To the best of our knowledge, within the framework where defaultable bonds are available as tradable assets, \citet{BoLiWangYang13} is the sole exception that includes intermediate consumption in the objective functional. However, their model is restricted to a market with a single defaultable bond and does not account for regime-switching dynamics, nor does it incorporate utility from intermediate or terminal wealth. 
Compared with \citet{BoLiWangYang13}, our formulation allows us to explore the delicate interplay and superposition effects arising from the trade-off between current consumption and the retention of wealth for future consumption or credit investment. 
Intuitively, current consumption generates immediate utility but reduces the wealth available for future investment, consumption, and terminal wealth. Conversely, as shown in \eqref{opti.strategy.inv} and \eqref{opti.strategy.cons}, the current optimal consumption rate and investment fractions are indirectly influenced by future investment, consumption, and terminal wealth through the common term $\varphi$, which, according to \eqref{continuation}, represents a scaled optimized value of the accumulated utilities from future immediate consumption, immediate wealth, and terminal wealth. {As demonstrated by numerical experiments, the optimal consumption and investment decisions confirm a trade-off between current consumption and wealth retention. For example, Figure~\ref{fig:numerical-eta} in Subsection~\ref{Subsec4.4} shows that an increase in the preference weight $\eta$ on current consumption leads to a substantial rise in the optimal consumption rate and a slight reduction in the fraction invested in risky assets. Another noteworthy observation is that, owing to differences in model specifications and problem formulations between \citet{Bo19} and our paper, we identify some novel phenomena. For instance, Figures~\ref{fig:numerical-asset2-intensity} and~\ref{fig:numerical-asset1-intensity} in Subsection~\ref{subsec4.3} reveal that when both defaultable assets remain alive, an increase in one asset’s default intensity reduces its own portfolio weight while slightly increasing that of the other asset. This contrasts with the simultaneous decline reported by \citet{Bo19}. The rationale for this divergence likely stems from two countervailing channels. Anticipated contagion exerts downward pressure on the other asset's portfolio weight, whereas reallocation away from the asset with heightened intensity exerts upward pressure; in our setting, the latter dominates.}

Second, our proof strategy differs from the techniques in \citet{JiaoKharroubiPham13}, \citet{Bo19}, and \citet{BoLiaoYu19} in two key respects. On one hand, while \citet{JiaoKharroubiPham13} employ backward induction to recursively solve a family of indexed quadratic BSDEs, our backward induction instead proceeds by solving regime-coupled nonlinear ODE systems across neighboring default states. In contrast to \citet{Bo19} and \citet{BoLiaoYu19}, where the recursion initiates from a linear all-default equation, our problem involves
minimization over the consumption rate $c$, which leaves a residual term of the form $y_i^p$ (with $p\in(0,1)$); see \eqref{dyn.n} and \eqref{G^n}. Since this term is not continuously differentiable at the origin, classical theory (e.g., Proposition 1.1 of Charter 3 in \cite{Smith95}) cannot be directly applied to establish the existence and uniqueness of positive solutions to the ODE \eqref{dyn.n}. To address this issue, we apply a truncation that removes the singularity at zero, establish a strictly positive lower bound independent of the truncation threshold, and subsequently verify that the truncation is indeed inactive along the solution. 
On the other hand, the verification arguments in \citet{Bo19} and \citet{JiaoKharroubiPham13}, are tailored to positive-valued power utility, exponential utility, respectively. In contrast, our objective incorporates both running and terminal utilities that are negative when \(\gamma>1\). Consequently, their limiting arguments do not carry over directly to our setting. Although the power utility of terminal wealth in \citet{BoLiaoYu19} is also negative-valued, a change of measure transforms the terminal wealth into an exponential functional of an accumulated running cost, and the verification arguments are then carried out under the transformed probability measure. Thus, the verification method of \citet{BoLiaoYu19} cannot be applied directly either.
In our analysis, for an arbitrary admissible strategy, we invoke the  maximal integrability condition to obtain $L^1$ convergence of the stopped terminal wealth, and apply the monotone convergence theorem to handle the running utility. For the candidate strategy, we derive feedback bounds, obtain an \(L^2\) estimate for the Brownian stochastic exponential, and employ Doob's inequality to establish the required maximal estimate for the negative-valued power utility of terminal wealth.


The remainder of the paper is organized as follows.
Section~\ref{sec:market-model} introduces the regime-switching credit
market, including the default-contagion mechanism, the jump-to-default
asset prices, and the investor's wealth dynamics.
Section~\ref{sec:cobb-douglas} formulates the portfolio-consumption
problem under Cobb--Douglas utility, derives the recursive HJB system,
establishes the existence and uniqueness of its positive classical
solution, characterizes the optimal feedback controls, and proves the
verification theorem. Section~\ref{sec:numerical} presents sensitivity analyses and comparisons across contagion settings, regimes, utilities, as well as before and after default.

\section{The Market Model}\label{sec:market-model}

Fix a finite investment horizon $T\in(0,\infty)$. We consider a credit market on $[0,T]$ consisting of $n\geq2$ defaultable risky assets and one riskless bond. For any positive integer $n$, $e_{n}$ denotes the $n$-dimensional column vector whose entries are all equal to one. For a vector $a$, $\operatorname{diag}(a)$ denotes the diagonal matrix with diagonal vector $a$. Throughout, $\|\cdot\|$ denotes the Euclidean norm for vectors and the corresponding operator norm for matrices, while $|\cdot|$ denotes the absolute value of a scalar.

\subsection{Default contagion model}\label{subsec:default-contagion}
The market is defined on a complete filtered probability space $(\Omega,\mathcal G,\mathbb G,\mathbb P)$, where $\mathbb G:=\mathbb F\vee\mathbb H_1\vee\mathbb H_2$ is augmented by all $\mathbb P$-null sets so as to satisfy the usual conditions. 

The reference filtration $\mathbb F=(\mathcal F_t)_{t\in[0,T]}$ is generated by a $d$-dimensional Brownian motion $W=(W_t^1,\ldots,W_t^d)_{t\in[0,T]}^\top$, where $d$ is a positive integer and $\top$ denotes transposition.
The filtration $\mathbb H_1=(\mathcal H_t^1)_{t\in[0,T]}$ contains all information about default events until the target horizon $T$. Let $H=(H_t^1,\ldots,H_t^n)_{t\in[0,T]}$ be the default indicator process, taking values in the finite default-state space $\mathcal S_H:=\{0,1\}^n$. In particular, $H_t^j=0$ means that asset $j$ is alive at time $t$, while $H_t^j=1$ means that it has defaulted. The default time of asset $j$ is therefore $\tau^j:=\inf\{t\in[0,T]:H_t^j=1\}$, with the convention $\inf\emptyset=\infty$. For $h=(h^1,\ldots,h^n)\in\mathcal S_H$, we write
$\widetilde h^j:=(h^1,\ldots,h^{j-1},1-h^j,h^{j+1},\ldots,h^n),$
which is the neighboring default state obtained by switching the $j$th component of $h$.
The filtration $\mathbb H_2=(\mathcal H_t^2)_{t\in[0,T]}$ is generated by a continuous-time Markov chain $I=(I_t)_{t\in[0,T]}$. Its finite state space is denoted by $\mathcal S_I:=\{1,\ldots,m\}$; and its generator is $Q=(q_{ij})_{i,j=1}^m$. The process $I$ is independent of $W$ and represents the observable regime of the credit market.
For two distinct regimes $i,j\in\mathcal S_I$, define the transition counting process
$K_t^{i,j}:=\sum_{0<s\leq t}\mathbf 1_{\{I_{s-}=i,I_s=j\}},$
and its compensated martingale
$\widetilde K_t^{i,j}:=K_t^{i,j}-\int_0^t q_{ij}\mathbf 1_{\{I_{s-}=i\}}\,ds.$

Using the random ingredients prepared above, we formulate the default contagion dynamic model which is described as the joint process $(I,H)=(I_t,H_t)_{t\in[0,T]}$ comprised of the regime-switching process and the default indicator process is a Markov process living in $\mathcal S_I\times\mathcal S_H$. Moreover, the default indicator process transits from a state $(H_{t-}^1,\ldots,H_{t-}^{j-1},H_{t-}^j,H_{t-}^{j+1},\ldots,H_{t-}^n)$ in which risky asset $j$ is alive $(H_{t-}^j=0)$ to the neighboring state $(H_t^1,\ldots,H_t^{j-1},H^{j}_t=1-H_{t-}^j,H_t^{j+1},\ldots,H_t^n)$ in which risky asset $j$ has defaulted at a random rate $\mathbf 1_{\{H_{t-}^j=0\}}\lambda_j(I_{t-},H_{t-})$. Here, we assume that the default intensity function $\lambda_j:\mathcal S_I\times\mathcal S_H\to(0,\infty)$ is bounded. This assumption provides uniform bounds on default intensities, which are used to control the compensator terms and keep the optimal portfolio weights uniformly below one; see \eqref{default.martingale} and \eqref{C.hat}.

We assume that $W$ is independent of $(I,H)$ and that $\mathbb G$ is the usual augmentation of their joint natural filtration. Almost surely, no two components of $(H^1,\ldots,H^n,I)$ jump simultaneously. Under these assumptions, $W$ remains a Brownian motion under $\mathbb G$, and the compensated counting processes associated with defaults and regime transitions are $\mathbb G$-martingales.

\subsection{Price dynamics of risky assets}\label{subsec:risky-asset-prices}
Building upon the default contagion model specified as in Section~\ref{subsec:default-contagion}, the price dynamics of risky asset $j$ is given by the following Black-Scholes model with default contagion under the framework of jump-to-default (see, \cite{Linetsky06}), for $j=1,\ldots,n$,
\begin{equation}\label{asset.price}
\left\{
\begin{array}{l}
S_t^j=(1-H_t^j)\widehat S_t^j,\\[1mm]
\displaystyle\frac{d\widehat S_t^j}{\widehat S_{t}^j}
=\bigl(\mu_j(I_{t})+\lambda_j(I_{t},H_{t})\bigr)dt
+\sum_{k=1}^d\sigma_{jk}(I_{t})dW_t^k,
\end{array}
\right.
\end{equation}
Here $\mu(i)=(\mu_1(i),\ldots,\mu_n(i))^\top$ is the regime-dependent expected return vector. The matrix $\sigma(i)=(\sigma_{jk}(i))_{1\leq j\leq n,1\leq k\leq d}$ is the regime-dependent $n\times d$ diffusion coefficient matrix. We assume throughout that $\sigma(i)\sigma(i)^\top$ is positive definite for every regime $i\in\mathcal S_I$.
We refer to $\widehat S^j=(\widehat S_t^j)_{t\in[0,T]}$ as the pre-default price for asset $j$. The dynamics of $\widehat S^j$ in \eqref{asset.price} indicates that the investor holding asset $j$ is compensated for the incurred default risk at the premium rate $\lambda_j(I_t,H_t)$. This is the additional rate of return that he is receiving, over the expected rate $\mu_j(I_t)$, to bear the default risk of the $j$th asset in the portfolio. The price representation $S_t^j$ for asset $j$ at time $t$ in \eqref{asset.price} implies that the price of the $j$th asset is given by the pre-default price $\widehat S_t^j$ up to $\tau^j$-, and jumps to zero at its default time $\tau^j$, where it remains forever afterwards. The Brownian term driven by $W$ represents continuous diffusion risk, whereas the default martingale introduced below generates the jump of the asset price to zero. The It\^{o}'s rule implies from \eqref{asset.price} that the price dynamics of asset $j$ follows:
\begin{equation}\label{asset.return}
{dS_t^j}
={S_{t-}^j}\left[\mu_j(I_t)dt
+\sum_{k=1}^d\sigma_{jk}(I_t)dW_t^k
-dM_t^j\right],\quad t\in[0,T].
\end{equation}
Here, the process $M^j=(M_t^j)_{t\in[0,T]}$ is a pure jump $\mathbb G$-martingale which is defined by
\begin{equation}\label{default.martingale}
M_t^j:=H_t^j-\int_0^{t\wedge\tau^j}\lambda_j(I_s,H_s)ds,
\quad \forall\,\, t\in[0,T].
\end{equation}

\begin{rem}
The assumption that $\sigma(i)\sigma(i)^\top$ is positive definite ensures that every nonzero investment vector $\pi$ satisfies $\|\sigma(i)^\top\pi\|^2\geq\delta_i\|\pi\|^2$, where $\delta_i>0$ is the smallest eigenvalue of $\sigma(i)\sigma(i)^\top$. This property is necessary for the subsequent analysis in three respects. First,  for $\gamma>1$, it makes the quadratic portfolio term in \eqref{H^k} strictly convex, so the portfolio minimizer is unique. Second, it yields the quadratic estimate \eqref{portfolio.quadratic.lower.bound}, which penalizes large portfolio positions and leads to the bounded-domain reduction in \eqref{G^k.2}. Third, it is inherited at every default state: once defaulted assets are removed, the covariance matrix of the surviving assets is a principal submatrix of $\sigma(i)\sigma(i)^\top$ and hence remains positive definite, so \eqref{portfolio.quadratic.lower.bound} continues to hold whenever risky assets survive.
\end{rem}

\subsection{Portfolio and consumption}\label{subsec:portfolio-consumption}
This section describes the portfolio and consumption strategy of one investor who will invest her wealth in the above $n$ risky assets and a riskless bond with interest rate $r(i)\geq0$ for $i\in\mathcal S_I$. Hence, the price dynamics of the riskless bond is $dB_t=r(I_t)B_tdt$ with $B_0=1$. For $j=1,\ldots,n$, let $\pi^j=(\pi_t^j)_{t\in[0,T]}$ be the $\mathbb G$-predictable fraction strategy for the $j$-th asset. It is assumed that the investor will not invest in the risky asset once it has defaulted. This yields that $1-\pi_t^\top e_n$ is the fraction strategy for the riskless bond at time $t$, where $\pi_t:=(\pi_t^j;j=1,\ldots,n)^\top$. In addition, we note that, since the price of the $j$-th asset jumps to zero when it defaults, the fraction of wealth held by the investor in this stock is zero after it defaults. In particular, it holds that $\pi_t^j=(1-H_{t-}^j)\pi_t^j$ for $j=1,\ldots,n$.
The jump-to-default constraint also requires $\pi_t^j<1$ before default, so we set the one-dimensional portfolio domain $\mathcal U:=(-\infty,1)$. This condition preserves strictly positive wealth when asset $j$ defaults. For a default state $h\in\mathcal S_H$, define the state-dependent portfolio domain
\[\mathcal U(h):=\{\pi\in\mathcal U^n:\pi^j=0\ \text{whenever }h^j=1\}.\]

Let $\mathbb R_+:=(0,\infty)$ and $\overline{\mathbb R}_+:=[0,\infty)$. The consumption control $c=(c_t)_{t\in[0,T]}$ is the consumption rate per unit of wealth and takes values in $\mathbb R_+$. For an admissible pair $(\pi,c)$ (the definition of admissibility is given in Definition \ref{def:admissible-strategies}), let $X_t^{\pi,c}$ denote the investor's wealth at time $t$. Combining the self-financing condition with \eqref{asset.return} yields
\begin{equation}\label{X.sde}
\frac{dX_t^{\pi,c}}{X_{t-}^{\pi,c}}
=\bigl(r(I_t)+\pi_t^\top(\mu(I_t)-r(I_t)e_n)-c_t\bigr)dt
+\pi_t^\top\sigma(I_t)dW_t
-\pi_t^\top dM_t,
\quad {X_0^{\pi,c}=x>0}.
\end{equation}
where $M_t=(M_t^j;j=1,\ldots,n)^\top$ for $t\in[0,T]$.
For later use, we also introduce the excess compensated return vector
\[
\widetilde\mu(i,h):=\mu(i)-r(i)e_n+\lambda(i,h),
\qquad (i,h)\in\mathcal S_I\times\mathcal S_H,
\]
where $\lambda(i,h):=(\lambda_1(i,h),\ldots,\lambda_n(i,h))^\top$. 
\section{Optimal Portfolio and Consumption under Cobb-Douglas Utility}\label{sec:cobb-douglas}

In this section, we study the optimal portfolio and consumption problem in the regime-switching credit market with default contagion introduced above. The investor derives utility from intermediate consumption and terminal wealth, and the Cobb-Douglas specification offers a concise economic description of the trade-off between present consumption and future financial security.

Define the Cobb--Douglas utility as
\begin{equation}\label{U(x)}
U_\eta(c,x):=\frac{(c^\eta x^{1-\eta})^{1-\gamma}}{1-\gamma},
\qquad \forall\,\,(c,x)\in\mathbb{R}_+^2,
\end{equation}
where $\gamma>1$ is the agent's risk aversion coefficient and $\eta\in(0,1]$ is interpreted as the weight parameter which reflects the agent's attitude to the consumption level and current wealth. Note that $u(x):=U_0(1,x)=\frac{x^{1-\gamma}}{1-\gamma}$ for $x\in\mathbb{R}_+$ is the classical power utility function which will have been widely used in the literature on the dynamical portfolio optimization.
Since $\gamma>1$, both $U_\eta$ and $u$ are negative-valued, increasing, and concave in the relevant wealth variable. For $0<\gamma<1$, the truncation and lower-bound arguments in the well-posedness proof require modification, but the verification step simplifies, because the utilities are then positive and Fatou's lemma applies directly to the localization limit in the verification inequality. The borderline case $\gamma=1$ lies outside this framework. There the power utilities degenerate, and the problem calls for logarithmic utilities together with a different value-function ansatz.

\begin{definition}\label{def:admissible-strategies}
A portfolio-consumption strategy $(\pi,c)$ is called admissible if it is a $\mathbb G$-predictable Markov feedback control of the following form
\begin{equation}\label{markov.controls}
\pi_t^j:=\pi^j(t,X_{t-}^{\pi,c},I_{t-},H_{t-}),\qquad
c_t:=c(t,X_{t-}^{\pi,c},I_{t-},H_{t-}),\qquad \forall\,\, t\in[0,T],\,\,j=1,\ldots,n,
\end{equation}
where $\pi^j:[0,T]\times\mathbb R_+\times\mathcal S_I\times\mathcal S_H\to\mathcal U$ and $c:[0,T]\times\mathbb R_+\times\mathcal S_I\times\mathcal S_H\to\mathbb R_+$ are measurable functions. We require $(\pi^1(t,x,i,h),\ldots,\pi^n(t,x,i,h))^\top\in\mathcal U(h)$ for every $(t,x,i,h)$ in this domain. In addition, for every initial state $(t,x,i,h)\in[0,T]\times\mathbb R_+\times\mathcal S_I\times\mathcal S_H$, the wealth equation \eqref{X.sde} starting from $X_t^{\pi,c}=x$ admits a unique positive strong solution on $[t,T]$, and,
\[
\mathbb E_{t,x,i,h}\left[\sup_{s\in[t,T]}(X_s^{\pi,c})^{1-\gamma}\right]<\infty,
\]
with $\mathbb{E}_{t,x,i,h}[\cdot]:=\mathbb{E}[\cdot\mid X_t=x,I_t=i,H_t=h]$.
Moreover, there exists a deterministic constant $\varepsilon_\pi>0$, possibly depending on the strategy $\pi$, such that
\[
1-\pi_s^j\geq\varepsilon_\pi,
\qquad s\in[0,T],\quad j=1,\ldots,n,
\qquad \text{a.s.}
\]
We denote by $\mathcal A$ the set of all portfolio-consumption strategies satisfying the above conditions.
\end{definition}

The preceding admissibility conditions ensure that the following criterion is well defined. For $(\pi,c)\in\mathcal A$, its performance from the initial state $(t,x,i,h)$ is
\begin{equation}\label{objective.cobb}
J(\pi,c;t,x,i,h):=\mathbb{E}_{t,x,i,h}\left[\int_t^T U_\eta(c_sX_s^{\pi,c},X_s^{\pi,c})ds+u(X_T^{\pi,c})\right].
\end{equation}
 The investor seeks a feedback strategy $(\pi^*,c^*)\in\mathcal A$ satisfying
\[
J(\pi^*,c^*;t,x,i,h)=\sup_{(\pi,c)\in\mathcal A}J(\pi,c;t,x,i,h).
\]

\subsection{Dynamic Programming Equations}\label{subsec:cobb-dpe}
In order to solve the stochastic control problem, we define the following value function given by
\begin{equation}\label{value.cobb}
V(t,x,i,h):=\sup_{(\pi,c)\in\mathcal A}J(\pi,c;t,x,i,h),\quad (t,x,i,h)\in[0,T]\times\mathbb{R}_+\times\mathcal S_I\times\mathcal S_H,
\end{equation}
where the objective functional $J(\pi,c;t,x,i,h)$ is defined by \eqref{objective.cobb}.


Using the dynamic program principle, the value function $V$ formally satisfies the following HJB equation:
\begin{align}\label{HJB.value.cobb}
0={}&\frac{\partial V(t,x,i,h)}{\partial t}+xr(i)\frac{\partial V(t,x,i,h)}{\partial x}
{\color{black}+\sum_{\ell\ne i}(V(t,x,\ell,h)-V(t,x,i,h))q_{i\ell}}\nonumber\\
&+\sup_{(\pi,c)\in\mathcal U(h)\times\mathbb{R}_+}\Bigg\{
U_\eta(cx,x)+x\frac{\partial V(t,x,i,h)}{\partial x}
\big(\pi^\top(I_n-\operatorname{diag}(h)){\color{black}\widetilde\mu(i,h)}-c\big)\nonumber\\
&\qquad+\frac12x^2\frac{\partial^2V(t,x,i,h)}{\partial x^2}
\left\|{\sigma(i)}^\top(I_n-\operatorname{diag}(h))\pi\right\|^2\nonumber\\
&\qquad+\sum_{j=1}^n\big(V(t,x(1-\pi^j),i,\widetilde h^j)-V(t,x,i,h)\big)
\lambda_j(i,h)(1-h^j)\Bigg\},
\end{align}
with terminal condition $V(T,x,i,h)=u(x)$ for all $(x,i,h)\in\mathbb{R}_+\times\mathcal S_I\times\mathcal S_H$. Here, \(I_n\) denotes the \(n\times n\) identity matrix.

The Cobb--Douglas criterion is homogeneous in wealth. This motivates the ansatz
\[
V(t,x,i,h)=\frac{x^{1-\gamma}}{1-\gamma}\varphi(t,i,h).
\]
Substituting this form into \eqref{HJB.value.cobb} and dividing by the common factor $x^{1-\gamma}/(1-\gamma)$ yields the recursive system
\begin{align}\label{HJB}
0=&\frac{\partial\varphi(t,i,h)}{\partial t}+(1-\gamma)r(i)\varphi(t,i,h)
{\color{black}+\sum_{\ell\in\mathcal S_I}\varphi(t,\ell,h)q_{i\ell}}\nonumber\\
&+\inf_{(\pi,c)\in\mathcal U(h)\times\mathbb{R}_+}
H_\eta(\pi,c;i,h,\varphi(t,i,\widetilde h^j);j=0,1,\ldots,n),
\end{align}
with terminal condition $\varphi(T,i,h)=1$ for all $(i,h)\in\mathcal S_I\times\mathcal S_H$. Above,we set $\widetilde h^0:=h$ and specify the Hamiltonian as, for any mapping $f:\mathcal{S}_{H}\rightarrow \mathbb{R},$
\begin{align}\label{H}
{\color{black}H_\eta\big(\pi,c;i,h,f(\widetilde h^j);j=0,1,\ldots,n\big)}:={}&c^{\eta(1-\gamma)}
+\left\{(1-\gamma)\big(\pi^\top(I_n-\operatorname{diag}(h)){\color{black}\widetilde\mu(i,h)}-c\big)\right.\nonumber\\
&\left.-\frac{\gamma(1-\gamma)}2\left\|{\sigma(i)}^\top(I_n-\operatorname{diag}(h))\pi\right\|^2\right\}f(h)\nonumber\\
&+\sum_{j=1}^n\big((1-\pi^j)^{1-\gamma}f(\widetilde h^j)-{f(h)}\big)
\lambda_j(i,h)(1-h^j).
\end{align}
Reversing the time by $\Bar{\varphi}(t,i,h):=\varphi(T-t,i,h)$. Then, Eq.~\eqref{HJB} becomes that
\begin{align}\label{HJB.2}
\frac{\partial\Bar{\varphi}(t,i,h)}{\partial t}
=&(1-\gamma)r(i)\Bar{\varphi}(t,i,h){\color{black}+\sum_{\ell\in\mathcal S_I}\Bar{\varphi}(t,\ell,h)q_{i\ell}}\nonumber\\
&+\inf_{(\pi,c)\in\mathcal U(h)\times\mathbb{R}_+}
H_\eta(\pi,c;i,h,\Bar{\varphi}(t,i,\widetilde h^j);j=0,1,\ldots,n),
\end{align}
with initial condition $\Bar{\varphi}(0,i,h)=1$ for all $(i,h)\in\mathcal S_I\times\mathcal S_H$.
To explore the well-posedness of the original recursive HJB equation \eqref{HJB.value.cobb} in classical sense, it is enough to study the well-posedness of the initial value recursive problem \eqref{HJB.2} using the following illustration of the route:

\begin{center}
\begin{tikzpicture}[
  routebox/.style={
    draw,
    align=center,
    inner xsep=4pt,
    inner ysep=3pt,
    text width=0.27\textwidth,
    font=\small
  },
  routearrow/.style={-{Latex[length=2mm]},thick}
]
\node[routebox] (initial) {Solution of Initial Value\\
Problem \eqref{HJB.2} $\Bar{\varphi}(t,i,h)$};
\node[routebox,right=8mm of initial] (terminal) {Solution of Terminal\\
Value Problem \eqref{HJB}\\
$\varphi(t,i,h)=\Bar{\varphi}(T-t,i,h)$};
\node[routebox,right=8mm of terminal] (hjb) {Solution of HJB Equation
\eqref{HJB.value.cobb}\\
$V(t,x,i,h)=\dfrac{x^{1-\gamma}}{1-\gamma}\varphi(t,i,h)$};
\draw[routearrow] (initial) -- (terminal);
\draw[routearrow] (terminal) -- (hjb);
\end{tikzpicture}
\end{center}

\subsection{The Well-posedness of Recursive Equation \eqref{HJB.2}}\label{subsec:cobb-well-posedness}
This subsection proves the existence and uniqueness of positive classical solutions to the recursive system \eqref{HJB.2}. For vectors $x,y\in\mathbb R^m$, we write $x\gg y$ if $x_i>y_i$ for every $i\in\mathcal S_I$. This order notation will be used to state strict positivity of the solutions.

The system is recursive because the equation at a default state $h$ depends on the neighboring states $\widetilde h^j$ reached after an additional default. Therefore the natural proof runs backward through default states: first solve the state in which all assets have defaulted, then use that solution to solve states with one surviving asset, and continue until the all-alive state is reached. For indices $j_1,\ldots,j_k$ that are distinct elements of $\{1,\ldots,n\}$, let $0^{j_1,\ldots,j_k}\in\mathcal S_H$ denote the vector obtained from the zero vector by setting the coordinates $j_1,\ldots,j_k$ equal to one. Thus $0$ denotes the all-alive state when $k=0$, while $0^{j_1,\ldots,j_n}=e_n^{\top}$ denotes the all-defaulted state.

In the sequel, we will use $(\pi^*,c^*)$ to represent the optimal (feedback) portfolio-consumption strategy. We distinguish two cases. The first and easier case is when the default state $h=e_n^{\top}$. Since all $n$ assets have defaulted, the decision maker does not allocate any wealth to them, i.e., $\pi^{1,*}=\cdots=\pi^{n,*}=0$. Thus, the portfolio and consumption space in this case is reduced to
$\mathcal A^{(n)}:=\{\pi^{(n)}:[0,T]\times\mathcal S_I\to\mathcal U^n;\ \pi^{(n)}\equiv0\}
\times\{c^{(n)}:[0,T]\times\mathcal S_I\to\mathbb R;\ c^{(n)}(t,i)>0\},$
and the Hamiltonian is reduced to $H_\eta^{(n)}(c;i,x):=c^{\eta(1-\gamma)}-x(1-\gamma)c.$
Since $\gamma>1$, it is not difficult to verify that $c\mapsto H_\eta^{(n)}$ is continuous and strictly convex on $\mathbb R_+$ satisfying $\lim_{c\rightarrow0+}H_\eta^{(n)}(c;i,x)=\lim_{c\rightarrow\infty}H_\eta^{(n)}(c;i,x)=+\infty$. Note that $H_\eta^{(n)}(1;i,x)=1-x(1-\gamma)<\infty$. Hence, there exists a unique optimum $c^*$ such that 
$\frac{\partial H_\eta^{(n)}(c^*;i,x)}{\partial c}=0.$
Thus, we derive that
\begin{eqnarray}\label{opti.c}
c^*(i,x)=\left(\frac{x}{\eta}\right)^{\frac{1}{\eta(1-\gamma)-1}}.
\end{eqnarray}
Plugging $h=e_n^\top$, $\pi^*\equiv0$ and $c^*$ given by \eqref{opti.c} into Eq.~\eqref{HJB.2}, and, letting $\Bar{\varphi}^{(n)}(t):=(\Bar{\varphi}(t,i,e_n^{\top});i\in\mathcal S_I)^{\top}$, we arrive at
\begin{eqnarray}\label{dyn.n}
\left\{
\begin{array}{ll}
    \frac{d}{dt} \Bar{\varphi}^{(n)}(t)&=\Lambda^{(n)}\Bar{\varphi}^{(n)}(t)+
    \Phi^{(n)}(\Bar{\varphi}^{(n)}(t)),
    \\
    \Bar{\varphi}^{(n)}(0)&=e_m,
\end{array}
\right.
\end{eqnarray}
where the coefficient matrix is given by
\begin{eqnarray}\label{A^n}
\Lambda^{(n)}\hspace{-0.3cm}&:=&\hspace{-0.3cm}
{\rm diag}\Big[((1-\gamma)r(i);\,i\in\mathcal S_I)\Big]+Q,
\end{eqnarray}
and the mapping $\Phi^{(n)}(y)=(\Phi^{(n)}_i(y);i\in\mathcal S_I)^{\top}$ for $y\in\mathbb{R}^m$ is defined as, for $i\in\mathcal S_I$,
\begin{eqnarray}\label{G^n}
\Phi^{(n)}_i(y):=
\Big(\eta^{-\frac{\eta(1-\gamma)}{\eta(1-\gamma)-1}}-(1-\gamma)\eta^{-\frac{1}{\eta(1-\gamma)-1}}\Big)y_i^{\frac{\eta(1-\gamma)}{\eta(1-\gamma)-1}}.
\end{eqnarray}

The following lemma settles this terminal node of the default recursion. Its role is to provide the starting point for the subsequent backward induction.
\begin{lem}\label{lem4.1}\label{lem4.2}\label{comparison}
The dynamical system \eqref{dyn.n} has a unique positive classical solution on $[0,T]$.
\end{lem}
\begin{proof}
Put $\alpha:=\eta(1-\gamma)<0$. For any $y,z\in\mathbb{R}^m$ satisfying $y,z\geq \epsilon e_m$ with $\epsilon>0$, the following estimate holds:
\begin{eqnarray}\label{G^n.lip}
\|\Phi^{(n)}(y)-\Phi^{(n)}(z)\|\leq C(\epsilon)\|y-z\|,
\end{eqnarray}
where $C(\epsilon)>0$ is a constant depending only on $\epsilon>0$. We next show the estimate \eqref{G^n.lip}. In fact, it suffices to prove that, for each $i\in\mathcal S_I$, $|\Phi_i^{(n)}(y)-\Phi_i^{(n)}(z)|\leq C(\epsilon)\|y-z\|$ for all $y,z\in\mathbb{R}^m$ satisfying $y,z\geq \epsilon e_m$. First of all, it follows from the mean value theorem and \eqref{G^n} that
\begin{eqnarray}
|\Phi_i^{(n)}(y)-\Phi_i^{(n)}(z)|
\hspace{-0.3cm}&=&\hspace{-0.3cm}
\Big(\eta^{-\frac{\alpha}{\alpha-1}}-(1-\gamma)\eta^{-\frac{1}{\alpha-1}}\Big)\Big|y_i^{\frac{\alpha}{\alpha-1}}-z_i^{\frac{\alpha}{\alpha-1}}\Big|
\nonumber\\
\hspace{-0.3cm}&\leq&\hspace{-0.3cm}
\Big(\eta^{-\frac{\alpha}{\alpha-1}}-(1-\gamma)\eta^{-\frac{1}{\alpha-1}}\Big)\frac{\alpha}{\alpha-1}\epsilon^{\frac{1}{\alpha-1}}|y_i-z_i|,\nonumber
\end{eqnarray}
This yields the validity of the estimate \eqref{G^n.lip}.
In what follows, for any constant $a\in(0,1]$, let us consider the following truncated system given by
\begin{eqnarray}\label{varphi_a^n}
\left\{
\begin{array}{ll}
     \frac{d}{dt}\Bar{\varphi}^{(n)}_a(t)=\Lambda^{(n)}\Bar{\varphi}^{(n)}_a(t)+\Phi^{(n)}_a(\Bar{\varphi}^{(n)}_a(t)),\quad t\in(0,T];  \\
     \Bar{\varphi}^{(n)}_a(0)=e_m,
\end{array}
\right.
\end{eqnarray}
where the matrix $\Lambda^{(n)}$ is given by \eqref{A^n} and the vector-valued function $\Phi^{(n)}_a(y)=(\Phi^{(n)}_{a,i}(y);i\in\mathcal S_I)^{\top}$ for $y\in\mathbb R^m$ is defined as
\begin{eqnarray}
\Phi^{(n)}_a(y):=\Phi^{(n)}(y\vee ae_m),\quad y\in\mathbb{R}^m,
\end{eqnarray}
with $y\vee z:=(y_i\vee z_i;i\in\mathcal S_I)^{\top}$ for $y,z\in\mathbb{R}^m$. Using the estimate \eqref{G^n.lip}, there exists a positive constant $C(a)$ depending on $a$ only such that, for all $y,z\in\mathbb R^m$,
\begin{eqnarray}\label{G_a^n.lip}
\|\Phi_{a}^{(n)}(y)-\Phi_{a}^{(n)}(z)\|\leq C(a)\|y\vee ae_m-z\vee ae_m\|\leq C(a)\|y-z\|,
\end{eqnarray}
It follows from \eqref{G_a^n.lip} that the mapping $y\mapsto \Lambda^{(n)}y+\Phi^{(n)}_a(y)$ is Lipschitz continuous. Hence, the system \eqref{varphi_a^n} has a unique classical solution $\Bar{\varphi}^{(n)}_a(t)=(\Bar{\varphi}^{(n)}_a(t,i);i\in\mathcal S_I)^{\top}$ on $[0,T]$.

We next show the positivity of the each element of $\Bar{\varphi}^{(n)}_a(t)$ for $t\in[0,T]$. To do it, consider the following linear system with unknown solution $\psi^{(n)}(t)=(\psi^{(n)}_i(t);i\in\mathcal S_I)^{\top}$ given by
\begin{eqnarray}
\label{4.16.6}
\left\{
\begin{array}{ll}
    \frac{d}{dt} {\psi}^{(n)}(t)&=\Lambda^{(n)}{\psi}^{(n)}(t),\quad t\in[0,T],
    \\
    {\psi}^{(n)}(0)&=e_m.
\end{array}
\right.
\end{eqnarray}
Recall the matrix $\Lambda^{(n)}$ specified by \eqref{A^n}. Then $\Lambda^{(n)}_{ij}=q_{ij}\geq0$ for all $i\neq j$, which implies that $y\mapsto\Lambda^{(n)}y$ is of type $K$ functions on $\mathbb R^m$ (c.f. \cite{Smith95}). Moreover, $y\mapsto\Lambda^{(n)}y$ is Lipschitz continuous on $\mathbb R^m$ with $\Lambda^{(n)}0=0$. Thus, one can verify by applying Lemma 4.1 in \cite{Bo19} that the system \eqref{4.16.6} has a unique positive classical solution $\psi^{(n)}(t)$ on $[0,T]$.

Setting
$
\epsilon^{(n)}:=\min_{i\in\mathcal S_I}\left\{\inf_{t\in[0,T]}\psi_i^{(n)}(t)\right\}.
$
Then, by the continuity of $t\mapsto\psi^{(n)}(t)$ and the fact that $\psi^{(n)}(t)\gg0$ for all $t\in[0,T]$, we have $\epsilon^{(n)}>0$. Note that, for $y\in\mathbb R^m$ and $i\in\mathcal S_I$,
\begin{eqnarray}\label{G^n>0}
\Phi^{(n)}_{a,i}(y)=\Big(\eta^{-\frac{\alpha}{\alpha-1}}-(1-\gamma)\eta^{-\frac{1}{\alpha-1}}\Big)(y_i\vee a)^{\frac{\alpha}{\alpha-1}}>0.
\end{eqnarray}
Using \eqref{G_a^n.lip}, \eqref{G^n>0}, the initial condition $\Bar{\varphi}_a^{(n)}(0)=\psi^{(n)}(0)=e_m$, the fact that $y\mapsto\Lambda^{(n)}y$ is Lipschitz continuous and of type $K$ on $\mathbb R^m$, and the comparison result given by Lemma 4.4 in \cite{Bo19}, one can claim that
\begin{equation}\label{positive.varphi.n}
\Bar{\varphi}_a^{(n)}(t)\geq\psi^{(n)}(t)\geq\epsilon^{(n)}e_m\gg 0,\quad \forall\,\, t\in[0,T].
\end{equation}
where it is recalled that the positive constant $\epsilon^{(n)}$ is independent of $a\in(0,1]$. So that, we choose $a\in(0,\epsilon^{(n)}\wedge1)$, and hence, it holds from \eqref{positive.varphi.n} that
$$\Phi^{(n)}_a(\Bar{\varphi}^{(n)}_a(t))=\Phi^{(n)}(\Bar{\varphi}^{(n)}_a(t)\vee ae_m)=\Phi^{(n)}(\Bar{\varphi}^{(n)}_a(t)),\quad \forall\,\, t\in[0,T].$$
Thus, $\Bar{\varphi}^{(n)}_a(t)\,(\geq\epsilon^{(n)}e_m)$ with $a\in(0,\epsilon^{(n)}\wedge1)$ solving the system \eqref{varphi_a^n} also satisfies the HJB system \eqref{dyn.n} on $[0,T]$. Moreover, by using \eqref{G^n.lip}, there is a unique positive classical solution to the HJB system \eqref{dyn.n}. Therefore, $\Bar{\varphi}^{(n)}_a(t)$ for $t\in[0,T]$ with $a\in(0,\epsilon^{(n)}\wedge1)$ is the unique positive classical solution for \eqref{dyn.n} on $[0,T]$. So far, the proof of the lemma is complete.
\end{proof}

We now move one step back in the recursion. Fix a state $h=0^{j_1,\ldots,j_k}$ with $0\leq k\leq n-1$. The assets indexed by $j_1,\ldots,j_k$ have already defaulted and hence receive zero portfolio weight. The surviving asset set is
$\mathcal J_k:=\{1,\ldots,n\}\setminus\{j_1,\ldots,j_k\}.$
For this state, 
$\Bar{\varphi}^{(k)}(t):=(\Bar{\varphi}(t,i,0^{j_1,\ldots,j_k});i\in\mathcal S_I)^\top,$
and let $\pi^{(k)}=(\pi^{(k),j};j\in\mathcal J_k)^\top$ collect only the portfolio weights of surviving assets. The coefficients restricted to the surviving assets are
\begin{itemize}
\item $\lambda^{(k)}(i)=(\lambda_j(i,0^{j_1,...,j_k});j\in\mathcal J_k)^{\top}$;
\item $\widetilde\mu^{(k)}(i)=
(\widetilde\mu_j(i,0^{j_1,...,j_k});j\in\mathcal J_k)^{\top}$;
\item {$\sigma^{(k)}(i)$ is the $(n-k)\times d$ submatrix of $\sigma(i)$ formed by the rows indexed by $\mathcal J_k$.}
\end{itemize}
These definitions simply restrict the original market coefficients to the assets that have not yet defaulted.
Accordingly, the admissible feedback values at this default state are restricted to the surviving assets, and the reduced control space is
\begin{eqnarray}
\mathcal A^{(k)}:=
\left\{(\pi^{(k)},c)
\in\mathcal U^{n-k}\times {{\mathbb R}_+}\right\},\nonumber
\end{eqnarray}
For a vector $x=(x_i;i\in\mathcal S_I)^\top$ representing the current values of $\Bar{\varphi}^{(k)}$, the Hamiltonian at the default state $0^{j_1,\ldots,j_k}$ becomes
\begin{eqnarray}\label{H^k}
H_\eta^{(k)}(\pi^{(k)},c;i,x)
\hspace{-0.3cm}&=&\hspace{-0.3cm}c^{\eta(1-\gamma)}-c(1-\gamma)x_i+(1-\gamma)
\left[{\pi^{(k)}}^{\top}{\color{black}\widetilde\mu^{(k)}(i)}-\frac{\gamma}{2}\Big\|{{\sigma^{(k)}(i)}^{\top}}{\pi^{(k)}}\Big\|^2\right]x_i
\nonumber\\
\hspace{-0.3cm}&&\hspace{-0.3cm}
+\sum_{j\notin \{j_1,...,j_k\}}\Big[(1-\pi^{(k),j})^{1-\gamma}\Bar{\varphi}^{(k+1),j}(t,i)-x_i\Big]\lambda_j^{(k)}(i).
\end{eqnarray}
Here, for each $j\notin\{j_1,\ldots,j_k\}$, we define $\Bar{\varphi}^{(k+1),j}(t,i)
:=\Bar{\varphi}(t,i,0^{j_1,\ldots,j_k,j}).$
It is the $i$-th component of the solution to the recursive HJB system \eqref{HJB.2} at the neighboring default state $0^{j_1,\ldots,j_k,j}$ reached when the surviving asset $j$ defaults.
Then, by using \eqref{H^k} and the fact that $\gamma>1$, one has that, for each $i\in\mathcal S_I$, $H_\eta^{(k)}(\pi^{(k)},c;i,x)$ is continuous and strictly convex in $(\pi^{(k)},c)\in\mathcal U^{n-k}\times{\mathbb{R}}_+$, which together with the expression of $H_\eta^{(k)}(\pi^{(k)},c;i,x)$ guarantees the existence a unique finite global minimum of $(\pi^{(k)},c)\mapsto H_\eta^{(k)}(\pi^{(k)},c;i,x)$. It follows from the first-order condition that, the optimal feedback strategy is given by
\begin{eqnarray}
c^*=c^*(i,x)=\Big(\frac{x_i}{\eta}\Big)^{\frac{1}{\eta(1-\gamma)-1}}.\nonumber
\end{eqnarray}
Hence, the corresponding HJB system \eqref{HJB.2} in this default state becomes that
\begin{eqnarray}\label{HJB.3}
\left\{
    \begin{array}{ll}
       \frac{d}{dt}\Bar{\varphi}^{(k)}(t)  &=\Lambda^{(k)}\Bar{\varphi}^{(k)}(t)+\Phi^{(k)}(t,\Bar{\varphi}^{(k)}(t)),\quad t\in[0,T];  \\
        \Bar{\varphi}^{(k)}(0) &=e_m,
    \end{array}
    \right.
\end{eqnarray}
where $\Lambda^{(k)}$ is the $m\times m$-dimensional matrix given by
\begin{eqnarray}\label{A^k}
\Lambda^{(k)}\hspace{-0.3cm}&:=&\hspace{-0.3cm}
{\rm diag}\bigg[{\color{black}(1-\gamma)r(i)-\sum_{j\notin \{j_1,...,j_k\}}\lambda^{(k)}_j(i)}
;\,i\in\mathcal S_I\bigg]+Q.
\end{eqnarray}
The coefficient $\Phi^{(k)}(t,y)=(\Phi^{(k)}_i(t,y);i\in\mathcal S_I)^{\top}$ for $(t,y)\in[0,T]\times\mathbb{R}^m$ is specified as
\begin{eqnarray}\label{G^k}
\Phi^{(k)}_i(t,y)
\hspace{-0.3cm}&:=&\hspace{-0.3cm}
\inf_{\pi^{(k)}\in\mathcal U^{n-k}}\left\{
\sum_{j\notin\{j_1,...,j_k\}}(1-\pi^{(k),j})^{1-\gamma}\Bar{\varphi}^{(k+1),j}(t,i)\lambda^{(k)}_j(i)+{\color{black}\Bar{H}^{(k)}(\pi^{(k)},i)y_i}\right\}
\nonumber\\
\hspace{-0.3cm}&&\hspace{1cm}
\left.
+\Big(\eta^{-\frac{\alpha}{\alpha-1}}-(1-\gamma)\eta^{-\frac{1}{\alpha-1}}\Big)y_i^{\frac{\alpha}{\alpha-1}}\right.
,
\end{eqnarray}
with $\alpha:=\eta(1-\gamma)<0$. 
In the above equation, the function
\begin{eqnarray}\label{bar.H}
\Bar{H}^{(k)}(\pi^{(k)},i):=(1-\gamma)
\left[{\pi^{(k)}}^{\top}{\color{black}\widetilde\mu^{(k)}(i)}-\frac{\gamma}{2}\Big\|{{\sigma^{(k)}(i)}^{\top}}{\pi^{(k)}}\Big\|^2\right].\nonumber
\end{eqnarray}

The next lemma is the key regularity estimate needed for the induction step. The nonlinearity $\Phi^{(k)}$ is not globally Lipschitz near the boundary $y_i=0$, because the power $y_i^{\alpha/(\alpha-1)}$ is singular there. However, once the solution is known to stay uniformly away from zero, the following local Lipschitz estimate is sufficient.
\begin{lem}\label{lem4.4}
Let $0\leq k\leq n-1$. Assume that the HJB system \eqref{HJB.2} has a unique positive classical solution $\Bar{\varphi}^{(k+1),j}(t)$ on $t\in[0,T]$ for $j\notin\{j_1,...,j_k\}$. Then, for any $y,z\in\mathbb{R}^m$ satisfying $y,z\geq \epsilon e_m$ with $\epsilon>0$, there exists a constant $C(\epsilon)>0$ depending on $\epsilon>0$ only such that
\begin{eqnarray}
\|\Phi^{(k)}(t,y)-\Phi^{(k)}(t,z)\|\leq C(\epsilon)\|y-z\|,\quad \forall\,\, t\in[0,T].\nonumber
\end{eqnarray}
\end{lem}
\begin{proof}
It follows from \eqref{G^n} and \eqref{G^k} that, for all $(t,i,y)\in[0,T]\times\mathcal S_I\times\mathbb R^m$,
\begin{eqnarray}
\Phi_i^{(k)}(t,y)
\hspace{-0.3cm}&=&\hspace{-0.3cm}
\inf_{\pi^{(k)}\in\mathcal U^{n-k}}\left\{
\sum_{j\notin\{j_1,...,j_k\}}(1-\pi^{(k),j})^{1-\gamma}\Bar{\varphi}^{(k+1),j}(t,i)\lambda^{(k)}_j(i)+\Bar{H}^{(k)}(\pi^{(k)},i)y_i\right\}+\Phi^{(n)}_i(y)
\nonumber\\
\hspace{-0.3cm}&:=&\hspace{-0.3cm}
G^{(k)}_i(t,y)+\Phi^{(n)}_i(y),\nonumber
\end{eqnarray}
where $\Phi^{(n)}(y)=(\Phi_i^{(n)}(y);i\in\mathcal S_I)^\top$ for $y\in\mathbb R^m$ is defined by \eqref{G^n}. It suffices to show that, for each $i\in\mathcal S_I$, $|G_i^{(k)}(t,y)-G_i^{(k)}(t,z)|\leq C(\epsilon)\|y-z\|$ for all $y,z\in\mathbb{R}^m$ satisfying $y,z\geq \epsilon e_m$ with $\epsilon>0$. Using the recursive assumption that $\Bar{\varphi}^{(k+1),j}(t)$ on $t\in[0,T]$ is the unique positive classical solution to \eqref{HJB.2} for $j\notin\{j_1,...,j_k\}$. Then, it is continuous on $[0,T]$, which yields the existence of a constant $C_0>0$ independent of $t$ such that $\sup_{t\in[0,T]}\|\Bar{\varphi}^{(k+1),j}(t)\|\leq C_0$ for $j\notin\{j_1,...,j_k\}$. Thus, by applying \eqref{G^k}, and the fact that $\Bar{H}^{(k)}(0,i)=0$ for all $i\in\mathcal S_I$, it follows that, for all $(t,y)\in[0,T]\times \mathbb{R}^m$,
\begin{eqnarray}\label{G^k<C_0}
G^{(k)}_i(t,y)
\hspace{-0.3cm}&=&\hspace{-0.3cm}
\inf_{\pi^{(k)}\in\mathcal U^{n-k}}
\left\{\sum_{j\notin\{j_1,...,j_k\}}(1-\pi^{(k),j})^{1-\gamma}\Bar{\varphi}^{(k+1),j}(t,i)\lambda^{(k)}_j(i)+\Bar{H}^{(k)}(\pi^{(k)},i)y_i\right\}
\nonumber\\
\hspace{-0.3cm}&\leq&\hspace{-0.3cm}
\sum_{j\notin\{j_1,...,j_k\}}\Bar{\varphi}^{(k+1),j}(t,i)\lambda^{(k)}_j(i)+\Bar{H}^{(k)}(0,i)y_i
\nonumber\\
\hspace{-0.3cm}&\leq&\hspace{-0.3cm}
C_0\sum_{j\notin\{j_1,...,j_k\}}\lambda^{(k)}_j(i).
\end{eqnarray}
By the positive-definiteness assumption on $\sigma(i)\sigma(i)^{\top}$, the principal submatrix $\sigma^{(k)}(i)\sigma^{(k)}(i)^{\top}$ is also positive definite. Hence, there exists a constant $\delta^{(k)}>0$ such that $\big\|{\sigma^{(k)}(i)}^{\top}\pi^{(k)}\big\|^2\geq \delta^{(k)}\|\pi^{(k)}\|^2$ for all $i\in\mathcal{S}_{I}$. Then, for any $\pi^{(k)}\in\mathcal U^{n-k}$,
\begin{eqnarray}\label{portfolio.quadratic.lower.bound}
-\frac{\gamma(1-\gamma)}{2}\Big\|{{\sigma^{(k)}(i)}^{\top}}{\pi^{(k)}}\Big\|^2\geq -\frac{\gamma(1-\gamma)}{2}\delta^{(k)} \|\pi^{(k)}\|^2.
\end{eqnarray}
On the other hand, for any $\pi^{(k)}\in\mathcal U^{n-k}$,
\begin{eqnarray}
(1-\gamma){\pi^{(k)}}^{\top}{\color{black}\widetilde\mu^{(k)}(i)}
\geq (1-\gamma)\|\pi^{(k)}\|{\color{black}\|\widetilde\mu^{(k)}(i)\|}
\geq {\color{black}C_1^{(k)}}(1-\gamma)\|\pi^{(k)}\|,\nonumber
\end{eqnarray}
with the constant {\color{black}$C_1^{(k)}:=\max_{i\in\mathcal S_I}\|\widetilde\mu^{(k)}(i)\|\geq0$}. Then, we have
\begin{eqnarray}
\Bar{H}^{(k)}(\pi^{(k)},i)\geq-\frac{\gamma(1-\gamma)}{2}\delta^{(k)} \|\pi^{(k)}\|^2+{\color{black}C_1^{(k)}} (1-\gamma)\|\pi^{(k)}\|.\nonumber
\end{eqnarray}
We next take the nonnegative constant defined as {\color{black}$C_2^{(k)}:=\frac{2C_1^{(k)}}{\gamma\delta^{(k)}}$}. For all $\pi^{(k)}\in\{\pi^{(k)}\in\mathcal U^{n-k};\|\pi^{(k)}\|\geq {\color{black}C_2^{(k)}}\}$, it holds that $\Bar{H}^{(k)}(\pi^{(k)},i)\geq0$ for all $i\in\mathcal S_I$. Then, for all $\pi^{(k)}\in\{\pi^{(k)}\in\mathcal U^{n-k};\|\pi^{(k)}\|\geq {\color{black}C_2^{(k)}}\}$ and all $y\geq \epsilon e_m$, we deduce that
\begin{eqnarray}
	\hspace{-0.3cm}&&\hspace{-0.3cm}
\sum_{j\notin\{j_1,...,j_k\}}(1-\pi^{(k),j})^{1-\gamma}\Bar{\varphi}^{(k+1),j}(t,i)\lambda^{(k)}_j(i)+\Bar{H}^{(k)}(\pi^{(k)},i)y_i\nonumber\\
	\hspace{-0.3cm}&&\hspace{-0.3cm}
\geq\Bar{H}^{(k)}(\pi^{(k)},i)y_i\geq\Big[-\frac{\gamma(1-\gamma)}{2}\delta^{(k)} \|\pi^{(k)}\|^2+{\color{black}C_1^{(k)}} (1-\gamma)\|\pi^{(k)}\|\Big]\epsilon.\nonumber
\end{eqnarray}
We take the following positive constant depending on $\epsilon>0$ as
$$C_3(\epsilon):=\frac{{\color{black}C_1^{(k)}}}{\gamma\delta^{(k)}}+\sqrt{\left(\frac{{\color{black}C_1^{(k)}}}{\gamma\delta^{(k)}}\right)^2-\frac{2}{\gamma(1-\gamma)\delta^{(k)}\epsilon}C_0\max_{i\in\mathcal{S}_I}\sum_{j\notin\{j_1,...,j_k\}}\lambda_j^{(k)}(i)}.$$
Then, for all $\pi^{(k)}\in\{\pi^{(k)}\in\mathcal U^{n-k};\|\pi^{(k)}\|\geq C_3(\epsilon)\}$ and all $y\geq \epsilon e_m$, it holds that
\begin{eqnarray}\label{sum+C_0>c_0}
\sum_{j\notin\{j_1,...,j_k\}}(1-\pi^{(k),j})^{1-\gamma}\Bar{\varphi}^{(k+1),j}(t,i)\lambda^{(k)}_j(i)+\Bar{H}^{(k)}(\pi^{(k)},i)y_i\geq C_0\sum_{j\notin\{j_1,...,j_k\}}\lambda_j^{(k)}(i).
\end{eqnarray}
Combining \eqref{G^k<C_0} and \eqref{sum+C_0>c_0}, it follows that
\begin{eqnarray}\label{G^k.2}
G_i^{(k)}(t,y)
\hspace{-0.3cm}&=&\hspace{-0.3cm}
\inf_{\pi^{(k)}\in\mathcal U^{n-k}}\left\{\sum_{j\notin\{j_1,...,j_k\}}(1-\pi^{(k),j})^{1-\gamma}\Bar{\varphi}^{(k+1),j}(t,i)\lambda^{(k)}_j(i)+\Bar{H}^{(k)}(\pi^{(k)},i)y_i\right\}
\nonumber\\
\hspace{-0.3cm}&=&\hspace{-0.3cm}
\inf_{\substack{\pi^{(k)}\in\mathcal U^{n-k}; \\ {\|\pi^{(k)}\|}\leq C_3(\epsilon)}}\left\{\sum_{j\notin\{j_1,...,j_k\}}(1-\pi^{(k),j})^{1-\gamma}\Bar{\varphi}^{(k+1),j}(t,i)\lambda^{(k)}_j(i)+\Bar{H}^{(k)}(\pi^{(k)},i)y_i\right\}.
\end{eqnarray}
Moreover, due to \eqref{G^k.2}, it follows that
\begin{eqnarray}\label{result}
G_i^{(k)}(t,y)
\hspace{-0.3cm}&=&\hspace{-0.3cm}
\inf_{\substack{\pi^{(k)}\in\mathcal U^{n-k}; \\ {\|\pi^{(k)}\|}\leq C_3(\epsilon)}}\left[\sum_{j\notin\{j_1,...,j_k\}}(1-\pi^{(k),j})^{1-\gamma}\Bar{\varphi}^{(k+1),j}(t,i)\lambda^{(k)}_j(i)\right.\nonumber\\
\hspace{-0.3cm}&&\hspace{2cm}
\left.+\Bar{H}^{(k)}(\pi^{(k)},i)z_i+\Bar{H}^{(k)}(\pi^{(k)},i)(y_i-z_i)\right]
\nonumber\\
\hspace{-0.3cm}&\leq&\hspace{-0.3cm}
\inf_{\substack{\pi^{(k)}\in\mathcal U^{n-k}; \\ {\|\pi^{(k)}\|}\leq C_3(\epsilon)}}\left[\sum_{j\notin\{j_1,...,j_k\}}(1-\pi^{(k),j})^{1-\gamma}\Bar{\varphi}^{(k+1),j}(t,i)\lambda^{(k)}_j(i)\right.
\nonumber\\
\hspace{-0.3cm}&&\hspace{2cm}
\left.+\Bar{H}^{(k)}(\pi^{(k)},i)z_i+C(\epsilon)|y_i-z_i|\right]
\nonumber\\
\hspace{-0.3cm}&=&\hspace{-0.3cm}
G_i^{(k)}(t,z)+C(\epsilon)|y_i-z_i|.
\end{eqnarray}
Here, the finite positive constant $C(\epsilon)=\max_{i\in\mathcal S_I}C^{(i)}(\epsilon)$, where, for $i\in\mathcal S_I$,
$$C^{(i)}(\epsilon):=\sup_{\substack{\pi^{(k)}\in\mathcal U^{n-k}; \\ {\|\pi^{(k)}\|}\leq C_3(\epsilon)}}|\Bar{H}^{(k)}(\pi^{(k)},i)|.$$
Note that the constant $C^{(i)}(\epsilon)$ given above is nonnegative and finite for all $i\in\mathcal S_I$. By using \eqref{result}, we derive that $|G_i^{(k)}(t,y)-G_i^{(k)}(t,z)|\leq C(\epsilon)\|y-z\|$ for any $y,z\in\mathbb{R}^m$ satisfying $y,z\geq \epsilon e_m$ with $\epsilon>0$. Thus, the proof of the lemma is complete.
\end{proof}

We can now state the induction step. Assuming that all neighboring states with one additional default have already been solved, Theorem~\ref{thm4.1} proves that the current state also admits a unique positive classical solution. 

\begin{thm}\label{thm4.1}
Let $k\in\{0,1,\ldots,n-1\}$. Assume that the HJB system \eqref{HJB.2} has a unique positive classical solution $\Bar{\varphi}^{(k+1),j}(t)$ on $[0,T]$ for all $j\notin\{j_1,...,j_k\}$. Then, there exists a unique positive classical solution $\Bar{\varphi}^{(k)}(t)$ on $t\in[0,T]$ to the HJB system \eqref{HJB.2} at the default state $h=0^{j_1,...,j_k}$, i.e., the HJB system \eqref{HJB.3} has a unique positive classical solution.
\end{thm}
\begin{proof}
We first consider the following dynamical system given by
\begin{eqnarray}\label{psi1}
\left\{
\begin{array}{ll}
     \frac{d}{dt}{\psi}(t)=\Lambda^{(k)}{\psi}(t)+{f}({\psi}(t)),\quad t\in(0,T];  \\
     {\psi}(0)=e_m,
\end{array}
\right.
\end{eqnarray}
Here, $\Lambda^{(k)}$ is the $m\times m$-dimensional matrix given by \eqref{A^k}, and the vector-valued function ${f}(x)=({f}_{i}(x);i\in\mathcal S_I)^{\top}$ is defined as
\begin{eqnarray}
{f}_{i}(x):=
\inf_{\pi^{(k)}\in\mathcal U^{n-k}}\Bar{H}^{(k)}(\pi^{(k)},i)x_i,\quad \forall\,\, i\in\mathcal S_I.\nonumber
\end{eqnarray}
Since the mapping $x\mapsto \Lambda^{(k)}x+{f}(x)$ is Lipschitz continuous, is of type $K$, and satisfies $\Lambda^{(k)}0+{f}(0)=0$, one can deduce from Lemma 4.1 in \cite{Bo19} that the dynamical system \eqref{psi1} has a unique positive classical solution $\psi(t)=(\psi(t,i);i\in\mathcal S_I)^{\top}\gg0$ on $[0,T]$.  Setting
$
\epsilon^{(k)}:=\min_{i\in\mathcal S_I}\left\{\inf_{t\in[0,T]}\psi(t,i)\right\}.
$
Then, by the continuity of $t\mapsto\psi(t)$ and the fact that $\psi(t)\gg0$ for all $t\in[0,T]$ as well as $\psi(0)=e_m$, we have $\epsilon^{(k)}\in(0,1]$. 

For any constant $a\in(0,\infty)$, we introduce the following truncated dynamical system specified as:
\begin{eqnarray}\label{psi11}
\left\{
\begin{array}{ll}
     \frac{d}{dt}{\psi}_1(t)=\Lambda^{(k)}{\psi}_1(t)+\Tilde{f}({\psi}_1(t))
     ,\quad t\in(0,T];  \\
     {\psi}_1(0)=e_m,
\end{array}
\right.
\end{eqnarray}
where the vector-valued function $\Tilde{f}(x)=(\Tilde{f}_{i}(x);i\in\mathcal S_I)^{\top}$ is given by
\begin{eqnarray}\label{tilde.f1}
\Tilde{f}_{i}(x):=
\Big(\eta^{-\frac{\alpha}{\alpha-1}}-(1-\gamma)\eta^{-\frac{1}{\alpha-1}}\Big)(x_i\vee a)^{\frac{\alpha}{\alpha-1}}+
f_i(x),\quad \forall\,\, i\in\mathcal S_I.
\end{eqnarray}
Letting $\alpha:=\eta(1-\gamma)<0$ and recalling that $1-\gamma<0$, one finds that the mapping $x\mapsto \Lambda^{(k)}x+\Tilde{f}(x)$ is Lipschitz continuous, and hence the system \eqref{psi11} has a unique classical solution ${\psi}_1(t)=(\psi_1(t,i);i\in\mathcal S_I)^{\top}$. On the other hand, because the mapping
$$x\mapsto \Big(\eta^{-\frac{\alpha}{\alpha-1}}-(1-\gamma)\eta^{-\frac{1}{\alpha-1}}\Big)(x_i\vee a)^{\frac{\alpha}{\alpha-1}}$$
is positive-valued and Lipschitz continuous, and the mapping $x\mapsto \Lambda^{(k)}x+{f}(x)$ is Lipschitz continuous and of type $K$, it follows from Lemma 4.4 in \cite{Bo19} or \cite{Smith95} that the unique classical solution $\psi_1=(\psi_1(t))_{t\in[0,T]}$ of Eq.~\eqref{psi11} dominates the unique positive classical solution $\psi=(\psi(t))_{t\in[0,T]}$ of the system \eqref{psi1}, i.e., $\psi_1(t)\geq\psi(t)\gg0$ on $[0,T]$.

In the sequel, we take the constant specified as
\begin{eqnarray}\label{a}
a\in(0,\epsilon^{(k)}\wedge1).
\end{eqnarray}
We then consider the following truncated dynamical system described as:
\begin{eqnarray}\label{psi1a}
\left\{
\begin{array}{ll}
     \frac{d}{dt}{\psi}_{1a}(t)=\Lambda^{(k)}{\psi}_{1a}(t)+\Tilde{f}_{a}({\psi}_{1a}(t)),\quad t\in(0,T];  \\
     {\psi}_{1a}(0)=e_m,
\end{array}
\right.
\end{eqnarray}
Above, the vector-valued function $\Tilde{f}_{a}(x)=(\Tilde{f}_{a,i}(x);i\in\mathcal S_I)^{\top}$ is given by
\begin{eqnarray}
\label{4.44.6}
\Tilde{f}_{a,i}(x):=
\Big(\eta^{-\frac{\alpha}{\alpha-1}}-(1-\gamma)\eta^{-\frac{1}{\alpha-1}}\Big)(x_i\vee a)^{\frac{\alpha}{\alpha-1}}+
\inf_{\pi^{(k)}\in\mathcal U^{n-k}}\Bar{H}^{(k)}(\pi^{(k)},i)(x_i\vee a).
\end{eqnarray}
It is seen that the function $x\mapsto\Lambda^{(k)}x+\Tilde{f}_a(x)$ is Lipschitz continuous. This implies that the system \eqref{psi1a} admits a unique classical solution $\psi_{1a}(t)=(\psi_{1a}(t,i);i\in\mathcal S_I)^{\top}$. Moreover, it follows from $\psi_1(t)\geq ae_m\gg0$ on $[0,T]$ that
\begin{eqnarray}
\Tilde{f}_{a}(\psi_1(t))=\Tilde{f}(\psi_1(t)\vee ae_m)=\Tilde{f}(\psi_1(t)),\nonumber
\end{eqnarray}
Thus, $\psi_1$ also solves \eqref{psi1a}. By uniqueness, $\psi_{1a}=\psi_1\geq ae_m\gg0$ on $[0,T]$, and hence $\psi_{1a}$ is the unique positive classical solution of \eqref{psi1a}.

We next consider the following truncated dynamical system given by
\begin{eqnarray}\label{psi2a}
\left\{
\begin{array}{ll}
     \frac{d}{dt}{\psi}_{2a}(t)=\Lambda^{(k)}{\psi}_{2a}(t)+\Phi^{(k)}_{a}(t,{\psi}_{2a}(t)),\quad t\in(0,T];  \\
     {\psi}_{2a}(0)=e_m,
\end{array}
\right.
\end{eqnarray}
where the vector-valued function $\Phi^{(k)}_{a}(t,x)=(\Phi^{(k)}_{a,i}(t,x);i\in\mathcal S_I)^{\top}$ for $(t,x)\in[0,T]\times\mathbb R^m$ is defined by
\begin{eqnarray}\label{tilde.f2}
\Phi^{(k)}_{a}(t,x):=\Phi^{(k)}(t,x\vee ae_m).
\end{eqnarray}
Above, the mapping $\Phi^{(k)}(t,x)=(\Phi^{(k)}_i(t,x);i\in\mathcal S_I)^{\top}$ is given by
\begin{eqnarray}
\Phi^{(k)}_i(t,x)
\hspace{-0.3cm}&=&\hspace{-0.3cm}
\inf_{\pi^{(k)}\in\mathcal U^{n-k}}\left\{
\sum_{j\notin\{j_1,...,j_k\}}(1-\pi^{(k),j})^{1-\gamma}\Bar{\varphi}^{(k+1),j}(t,i)\lambda^{(k)}_j(i)+\Bar{H}^{(k)}(\pi^{(k)},i)x_i\right\}
\nonumber\\
\hspace{-0.3cm}&&\hspace{1cm}
\left.
+\Big(\eta^{-\frac{\alpha}{\alpha-1}}-(1-\gamma)\eta^{-\frac{1}{\alpha-1}}\Big)x_i^{\frac{\alpha}{\alpha-1}}\right.
\end{eqnarray}
By Lemma \ref{lem4.4}, one can verify that, for all $t\in[0,T]$ and $y,z\in\mathbb R^m$,
\begin{eqnarray}\label{G^k_a.lip}
\|\Phi^{(k)}_a(t,y)-\Phi^{(k)}_a(t,z)\|\leq C(a)\|(y\vee ae_m)-(z\vee ae_m)\|\leq C(a)\|y-z\|,
\end{eqnarray}
where the positive constant $C(a)$ depends only on $a>0$.
By applying \eqref{G^k_a.lip}, the mapping $x\mapsto \Lambda^{(k)}x+\Phi^{(k)}_{a}(t,x)$ is Lipschitz continuous uniformly in $t\in[0,T]$, and hence the system \eqref{psi2a} has a unique classical solution ${\psi}_{2a}(t)=({\psi}_{2a}(t,i);i\in\mathcal S_I)^{\top}$ on $[0,T]$. On the other hand, combine \eqref{4.44.6} and \eqref{tilde.f2} to have that, for all $i\in\mathcal S_I$,
\begin{eqnarray}
\Phi^{(k)}_{a,i}(t,x)
\hspace{-0.3cm}&=&\hspace{-0.3cm}
\inf_{\pi^{(k)}\in\mathcal U^{n-k}}\left\{
\sum_{j\notin\{j_1,...,j_k\}}(1-\pi^{(k),j})^{1-\gamma}\Bar{\varphi}^{(k+1),j}(t,i)\lambda^{(k)}_j(i)+\Bar{H}^{(k)}(\pi^{(k)},i)(x_i\vee a)\right\}
\nonumber\\
\hspace{-0.3cm}&&\hspace{1cm}
\left.
+\Big(\eta^{-\frac{\alpha}{\alpha-1}}-(1-\gamma)\eta^{-\frac{1}{\alpha-1}}\Big)(x_i\vee a)^{\frac{\alpha}{\alpha-1}}\right.
\nonumber\\
\hspace{-0.3cm}&\geq&\hspace{-0.3cm}
\inf_{\pi^{(k)}\in\mathcal U^{n-k}}\Bar{H}^{(k)}(\pi^{(k)},i)(x_i\vee a)+\Big(\eta^{-\frac{\alpha}{\alpha-1}}-(1-\gamma)\eta^{-\frac{1}{\alpha-1}}\Big)(x_i\vee a)^{\frac{\alpha}{\alpha-1}}
\nonumber\\
\hspace{-0.3cm}&=&\hspace{-0.3cm}
\Tilde{f}_{a,i}(x).\nonumber
\end{eqnarray}
This, by Lemma 4.4 in \cite{Bo19} with the initial condition $\psi_{2a}(0)=\psi_{1a}(0)=e_m$, $x\mapsto \Lambda^{(k)}x+\Tilde{f}_a(x)$ is Lipschitz continuous and of type $K$, and $x\mapsto \Phi^{(k)}_{a,i}(t,x)-\Tilde{f}_{a,i}(x)$ is positive-valued and Lipschitz continuous uniformly in $t\in[0,T]$. This yields that $\psi_{2a}(t)\geq\psi_{1a}(t)=\psi_{1}(t)\geq a e_m\gg 0$ for all $t\in[0,T]$.
Hence, it holds that 
\begin{eqnarray}
\Phi^{(k)}_a(t,\psi_{2a}(t))=\Phi^{(k)}(t,\psi_{2a}(t)\vee ae_m)=\Phi^{(k)}(t,\psi_{2a}(t)),\quad \forall\,\, t\in[0,T].\nonumber
\end{eqnarray}
In this vein, the function $\psi_{2a}(t)\,(\geq ae_m)$ with $a$ given by \eqref{a} solving the HJB system \eqref{psi2a}, also solves the HJB system \eqref{HJB.3} on $[0,T]$. Furthermore, by Lemma \ref{lem4.4} (with $\epsilon$ chosen to be less than or equal to $a$ given as in \eqref{a}), there is a unique positive classical solution to the HJB system \eqref{HJB.3}. Hence, $\psi_{2a}(t)\geq ae_m$ is the unique positive classical solution to the HJB system \eqref{HJB.3} on $[0,T]$. This completes the proof of the theorem.
\end{proof}

Theorem~\ref{thm4.1} gives the unique positive classical solution of the initial-value system \eqref{HJB.2}. Returning to the original time variable, $\varphi^{(k)}(t):=\Bar{\varphi}^{(k)}(T-t)$ solves \eqref{HJB} at the default state $h=0^{j_1,\ldots,j_k}$. The lower bound $\overline{\varphi}^{(k)}(t)\geq ae_m$ allows us to apply \eqref{G^k.2} with $\epsilon=a$, so that the portfolio minimization in \eqref{G^k} can be restricted to the bounded set $\{\pi^{(k)}\in\mathcal U^{n-k}:\|\pi^{(k)}\|\leq C_3(a)\}$. 
Thus, for all $i\in\mathcal S_I$,
\begin{align*}
G_i^{(k)}(t,\bar{\varphi}^{(k)}(t))
=\inf_{\substack{\pi^{(k)}\in\mathcal U^{n-k}\\ {\|\pi^{(k)}\|}\leq C_3(a)}}\Bigg\{
&\sum_{j\notin\{j_1,...,j_k\}}(1-\pi^{(k),j})^{1-\gamma}
\Bar{\varphi}^{(k+1),j}(t,i)\lambda^{(k)}_j(i)\\
&+\Bar{H}^{(k)}(\pi^{(k)},i)\bar{\varphi}^{(k)}(t,i)\Bigg\}.
\end{align*}
Since $\gamma>1$, it is not difficult to verify that, for any $i\in\mathcal S_I$, the objective function on the right-hand side of the above equation is continuous and strictly convex in $\pi^{(k)}\in\mathcal U^{n-k}$. Therefore, for the surviving assets, define the candidate optimal feedback investment by
\begin{align}
\pi^{(k),*}(t,i,h)
:=\operatorname*{arg\,min}_{\substack{\pi^{(k)}\in\mathcal U^{n-k}\\ {\|\pi^{(k)}\|}\leq C_3(a)}}
\Bigg\{\sum_{\ell\in\mathcal J_k}
(1-\pi^{(k),\ell})^{1-\gamma}
\varphi(t,i,\widetilde h^\ell)\lambda_\ell^{(k)}(i)
+\overline H^{(k)}(\pi^{(k)},i)\varphi(t,i,h)\Bigg\},
\label{opti.strategy.inv}
\end{align}
and set $\pi^{j,*}(t,i,h)=0$ for $j\in\{j_1,\ldots,j_k\}$ and $\pi^{j,*}(t,i,h)=\pi^{(k),j,*}(t,i,h)$ for $j\in\mathcal J_k$. When $k=n$, all risky-asset holdings are zero.
The unique candidate optimal consumption feedback is 
\begin{equation}\label{opti.strategy.cons}
c^*(t,i,h)=\left(\frac{\varphi(t,i,h)}{\eta}\right)^{\frac{1}{\eta(1-\gamma)-1}}.
\end{equation}
These feedback functions uniquely minimize the Hamiltonian. Their admissibility and optimality are established in the verification theorem below.

The next estimate concerns the optimal investment strategy characterized by \eqref{opti.strategy.inv}. 
\begin{lem}\label{lem4.5}
At any default state $h=0^{j_1,...,j_k}$ for $0\leq k\leq n$, the candidate portfolio feedback $\pi^{*}$  characterized by \eqref{opti.strategy.inv} satisfies the following estimates for some constants $C_{\pm}>0$:
$$
\left\{
\begin{array}{l}
\displaystyle\sup\limits_{(t,i,h)\in[0,T]\times\mathcal S_I\times\mathcal S_H}\|\pi^{*}(t,i,h)\|\leq C_{+},\\[2mm]
\displaystyle\min_{j\in\{1...,n\}}\inf\limits_{(t,i,h)\in[0,T]\times\mathcal S_I\times\mathcal S_H}(1-\pi^{j,*}(t,i,h))\geq C_{-}.
\end{array}
\right.
$$
\end{lem}
\begin{proof}
For the case with $k=n$, the conclusion is trivial since $\pi^{1,*}=\cdots=\pi^{n,*}=0$. For the case where $0\leq k\leq n-1$, using similar arguments as those in the proof of Lemma \ref{lem4.4}, one can check that, for all $(t,i,h)\in[0,T]\times\mathcal S_I\times\mathcal S_H$,
\begin{eqnarray}\label{C3}
\|\pi^{(k),*}(t,i,h)\|
\hspace{-0.3cm}&\leq&\hspace{-0.3cm}
\frac{{C_{\mu}}}{\gamma\delta_{\min}}+\sqrt{\left(\frac{{C_{\mu}}}{\gamma\delta_{\min}}\right)^2-\frac{2n\lambda_{\max}C_{\max}}{\gamma(1-\gamma)\delta_{\min} C_{\min}}}
:=C_{+}\in(0,\infty),
\end{eqnarray}
where $C_{\max}$, $C_{\min}$, {$C_{\mu}$} and $\lambda_{\max}$ are finite positive constants defined by
$$C_{\max}:=\sup\limits_{(t,i,h)\in[0,T]\times\mathcal S_I\times\mathcal S_H}|\varphi(t,i,h)|,\quad C_{\min}:=\inf\limits_{(t,i,h)\in[0,T]\times\mathcal S_I\times\mathcal S_H}|\varphi(t,i,h)|$$
$${C_{\mu}
	:=
	\max_{(i,h)\in\mathcal S_I\times\mathcal S_H}
	\left\|
	\bigl(I_n-\operatorname{diag}(h)\bigr)\widetilde\mu(i,h)
	\right\|},\quad \lambda_{\max}:=\max_{j\in\{1,...,n\}}\max_{(i,h)\in\mathcal S_I\times\mathcal S_H}\lambda_j(i,h),$$
and $\delta_{\min}:=\min_{i\in\mathcal S_I}\delta_i\in(0,\infty)$, where $\delta_i$ is the smallest eigenvalue of the positive definite matrix $\sigma(i)\sigma(i)^\top$. 

In addition, for $0\leq k\leq n-1$, the candidate investment feedback vector $\pi^{(k),*}$  defined by \eqref{opti.strategy.inv} satisfies
\begin{eqnarray}
\hspace{-0.3cm}&&\hspace{-0.3cm}
\sum_{\ell\notin\{j_1,...,j_k\}}(1-\pi^{(k),\ell,*})^{1-\gamma}{{\varphi}(t,i,\widetilde{h}^{\ell})}\lambda^{(k)}_{\ell}(i)
+
\Bar{H}^{(k)}(\pi^{(k),*},i){\varphi(t,i,h)}
\nonumber\\
\hspace{-0.3cm}&=&\hspace{-0.3cm}
\inf_{\substack{\pi^{(k)}\in\mathcal U^{n-k}; \\ {\|\pi^{(k)}\|}\leq C_3(a)}}\left\{
\sum_{\ell\notin\{j_1,...,j_k\}}(1-\pi^{(k),\ell})^{1-\gamma}{{\varphi}(t,i,\widetilde{h}^{\ell})}\lambda^{(k)}_{\ell}(i)+\Bar{H}^{(k)}(\pi^{(k)},i){\varphi(t,i,h)}\right\}
\nonumber\\
\hspace{-0.3cm}&\leq&\hspace{-0.3cm}
\sum_{\ell\notin\{j_1,...,j_k\}}{{\varphi}(t,i,\widetilde{h}^{\ell})}\lambda^{(k)}_{\ell}(i)+\Bar{H}^{(k)}(0,i){\varphi(t,i,h)},
\nonumber
\end{eqnarray}
This yields that
\begin{eqnarray}\label{3.41}
\hspace{-0.3cm}&&\hspace{-0.3cm}
\sum_{\ell\notin\{j_1,...,j_k\}}(1-\pi^{(k),\ell,*})^{1-\gamma}{{\varphi}(t,i,\widetilde{h}^{\ell})}\lambda^{(k)}_{\ell}(i)
\nonumber\\
\hspace{-0.3cm}&\leq&\hspace{-0.3cm}
\sum_{\ell\notin\{j_1,...,j_k\}}{{\varphi}(t,i,\widetilde{h}^{\ell})}\lambda^{(k)}_{\ell}(i)+{\varphi(t,i,h)}\big[\Bar{H}^{(k)}(0,i)-\Bar{H}^{(k)}(\pi^{(k),*},i)\big]
\nonumber\\
\hspace{-0.3cm}&=&\hspace{-0.3cm}
\sum_{\ell\notin\{j_1,...,j_k\}}{{\varphi}(t,i,\widetilde{h}^{\ell})}\lambda^{(k)}_{\ell}(i)-{\varphi(t,i,h)}(1-\gamma)
\left[{\pi^{(k),*}}^{\top}\widetilde\mu^{(k)}(i)-\frac{\gamma}{2}\Big\|{{\sigma^{(k)}(i)}^{\top}}{\pi^{(k),*}}\Big\|^2\right].
\end{eqnarray}
It follows from \eqref{C3} and \eqref{3.41} that
\begin{eqnarray}
\label{4.53}
\sum_{\ell\notin\{j_1,...,j_k\}}(1-\pi^{(k),\ell,*})^{1-\gamma}C_{\min}\lambda_{\min}
\hspace{-0.3cm}&\leq&\hspace{-0.3cm}
\sum_{\ell\notin\{j_1,...,j_k\}}(1-\pi^{(k),\ell,*})^{1-\gamma}{{\varphi}(t,i,\widetilde{h}^{\ell})}\lambda^{(k)}_{\ell}(i)
\nonumber\\
\hspace{-0.3cm}&\leq&\hspace{-0.3cm}
n\lambda_{\max}C_{\max}+C_{\max}|1-\gamma|
\left[C_+{\color{black}C_{\mu}}+\frac{\gamma}{2}\delta_{\max}^2C_+^2\right],
\end{eqnarray}
where $\lambda_{\min}:=\min_{j\in\{1,...,n\}}\min_{(i,h)\in\mathcal S_I\times\mathcal S_H}\lambda_j(i,h)$ and {$\delta_{\max}:=\max_{i\in\mathcal S_I}\|\sigma(i)\|$} are finite positive constants. In light of \eqref{4.53}, one has that
\begin{eqnarray}\label{C.hat}
\max_{\ell\notin\{j_1,...,j_k\}}(1-\pi^{(k),\ell,*})^{1-\gamma}
\hspace{-0.3cm}&\leq&\hspace{-0.3cm}
(C_{\min}\lambda_{\min})^{-1}\Big[n\lambda_{\max}C_{\max}+C_{\max}|1-\gamma|\left[C_+{\color{black}C_{\mu}}+\frac{\gamma}{2}{\delta}_{\max}^2C_+^2\right]\Big]
\nonumber\\
\hspace{-0.3cm}&:=&\hspace{-0.3cm} C_{-}^{1-\gamma},
\end{eqnarray}
which together with $1-\gamma<0$ yields
\begin{eqnarray}
\min_{\ell\notin\{j_1,...,j_k\}}(1-\pi^{(k),\ell,*})\geq C_{-},\nonumber
\end{eqnarray}
Here, the positive constant $C_{-}$ is independent of $(t,i,h)\in[0,T]\times\mathcal S_I\times\mathcal S_H$. Thus, the proof of the lemma is complete.
\end{proof}


We conclude the analysis by verifying that the unique positive classical solution constructed above is indeed the value function of the original control problem.
\begin{thm}[Verification Result]\label{verify.thm}
At any default state $h=0^{j_1,...,j_k}$ for $k=0,1,...,n$, let $\varphi(t,h)=(\varphi(t,i,h);i\in\mathcal S_I)^{\top}$ with $(t,h)\in[0,T]\times\mathcal S_H$ be the unique positive classical solution of the recursive HJB system \eqref{HJB}. Let $({\pi}^{*},{c}^*)$ be the candidate feedback pair defined by \eqref{opti.strategy.inv} and \eqref{opti.strategy.cons}. Then, we have
\begin{itemize}
\item[{\rm (\romannumeral1)}] For $(t,x,i,h)\in[0,T]\times\mathbb{R}_+\times\mathcal S_I\times\mathcal S_H,$ and any admissible strategy $(\pi,c)\in\mathcal A$, it holds that
\begin{eqnarray}
\frac{x^{1-\gamma}}{1-\gamma}\varphi(t,i,h)\geq\mathbb{E}_{t,x,i,h}\bigg[\int_t^TU_\eta(c_sX^{\pi,c}_s,X^{\pi,c}_s)ds+u(X^{\pi,c}_T)\bigg].\nonumber
\end{eqnarray}
\item[{\rm(\romannumeral2)}]
The portfolio-consumption strategy $(\pi^*,c^*)$  {defined by \eqref{opti.strategy.inv} and \eqref{opti.strategy.cons}} is admissible. Moreover, the value function $V(t,x,i,h)$ for $(t,x,i,h)\in[0,T]\times\mathbb{R}_+\times\mathcal S_I\times\mathcal S_H$ admits the following representation
\begin{eqnarray}\label{continuation}
V(t,x,i,h)=\mathbb{E}_{t,x,i,h}\bigg[\int_t^TU_\eta(c^*_sX^{\pi^*,c^*}_s,X^{\pi^*,c^*}_s)ds+u(X^{\pi^*,c^*}_T)\bigg]=\frac{x^{1-\gamma}}{1-\gamma}\varphi(t,i,h).
\end{eqnarray}
\end{itemize}
\end{thm}
\begin{proof}

Fix an admissible strategy $(\pi,c)\in\mathcal A$ and an initial state $(t,x,i,h)$. For each $\nu\in\mathbb N$, define
\[
\tau_\nu^t
:=
\inf\left\{s\in(t,T]:
(X_s^{\pi,c})^{1-\gamma}>\nu
\ \text{or}\
\int_t^s
\big\|\sigma(I_{v-})^\top
\big(I_n-\operatorname{diag}(H_{v-})\big)\pi_v\big\|^2dv>\nu
\right\},
\]
where $\inf\emptyset=+\infty$. Since $(\pi,c)\in\mathcal{A}$, we have
\[
\sup_{s\in[t,T]}(X_s^{\pi,c})^{1-\gamma}<\infty
\quad\text{and}\quad
\int_t^T
\big\|\sigma(I_{v-})^\top
\big(I-\operatorname{diag}(H_{v-})\big)\pi_v\big\|^2dv<\infty,
\qquad\text{a.s.}
\]
Consequently, $T\wedge\tau_\nu^t\uparrow T$ almost surely as $\nu\to\infty$. It\^o's formula applied on $[t,T\wedge\tau_\nu^t]$ yields
\begin{eqnarray}\label{verify.equa}
\hspace{-0.3cm}&&\hspace{-0.3cm}
\frac{1}{1-\gamma}\left(X^{\pi,c}_{T\wedge\tau_\nu^t}\right)^{1-\gamma}\varphi(T\wedge\tau_\nu^t,I_{T\wedge\tau_\nu^t},H_{T\wedge\tau_\nu^t})
\nonumber\\
\hspace{-0.3cm}&=&\hspace{-0.3cm}
\frac{x^{1-\gamma}}{1-\gamma}\varphi(t,i,h)
+\int_t^{T\wedge\tau_\nu^t}\frac{1}{1-\gamma}\left(X^{\pi,c}_{v-}\right)^{1-\gamma}\bigg[\frac{\partial \varphi(v,I_{v-},H_{v-})}{\partial v}\bigg]dv
\nonumber\\
\hspace{-0.3cm}&&\hspace{-0.3cm}
+\int_t^{T\wedge\tau_\nu^t}\left(X^{\pi,c}_{v-}\right)^{1-\gamma}\Big(r(I_{v-})+\pi^{\top}_v(I_n-\operatorname{diag}(H_{v-}))\widetilde\mu(I_{v-},H_{v-})-c_v
\nonumber\\
\hspace{-0.3cm}&&\hspace{-0.3cm}
-\frac{\gamma}{2}\Big\|{\sigma(I_{v-})}^{\top}(I_n-\operatorname{diag}(H_{v-}))\pi_v\Big\|^2\Big)\varphi(v,I_{v-},H_{v-})dv
\nonumber\\
\hspace{-0.3cm}&&\hspace{-0.3cm}
+\int_t^{T\wedge\tau_\nu^t}\left(X^{\pi,c}_{v-}\right)^{1-\gamma}\varphi(v,I_{v-},H_{v-})\pi_v^{\top}(I_n-\operatorname{diag}(H_{v-})){\sigma(I_{v-})dW_v}
\nonumber\\
\hspace{-0.3cm}&&\hspace{-0.3cm}
+\int_t^{T\wedge\tau_\nu^t}\sum_{j=1}^n\frac{1}{1-\gamma}\left(X^{\pi,c}_{v-}\right)^{1-\gamma}\Big[(1-\pi^j_v)^{1-\gamma}\varphi(v,I_{v-},\widetilde{H}^j_{v-})-\varphi(v,I_{v-},{H}_{v-})\Big]dH^j_{v}
\nonumber\\
\hspace{-0.3cm}&&\hspace{-0.3cm}
+\int_t^{T\wedge\tau_\nu^t}\sum_{\ell\neq I_{v-}}\frac{1}{1-\gamma}\left(X^{\pi,c}_{v-}\right)^{1-\gamma}\Big[\varphi(v,\ell,{H}_{v-})-\varphi(v,I_{v-},{H}_{v-})\Big]dK^{I_{v-},\ell}_v
\nonumber\\
\hspace{-0.3cm}&=&\hspace{-0.3cm}
\frac{x^{1-\gamma}}{1-\gamma}\varphi(t,i,h)
+\int_t^{T\wedge\tau_\nu^t}\frac{1}{1-\gamma}\left(X^{\pi,c}_{v-}\right)^{1-\gamma}\mathcal{L}(\pi_{v},c_v;v,I_{v-},H_{v-})dv+N^{\pi,c}_{T\wedge\tau_\nu^t}-N^{\pi,c}_t.
\end{eqnarray}
Here, for $(\pi,c)\in\mathcal{U}(h)\times(0,\infty)$ and $(t,i,h)\in[0,T]\times\mathcal S_I\times\mathcal S_H$, the coefficient is defined by
\begin{eqnarray}
\mathcal{L}(\pi,c;t,i,h)
\hspace{-0.3cm}&:=&\hspace{-0.3cm}
\frac{\partial\varphi(t,i,h)}{\partial t}+(1-\gamma)r(i)\varphi(t,i,h)
+\sum_{\ell\in\mathcal S_I}\varphi(t,\ell,h)q_{i\ell}
\nonumber\\
\hspace{-0.3cm}&&\hspace{-2.5cm}
+\Big[(1-\gamma)\Big({\pi}^{\top}(I_n-\operatorname{diag}(h))\widetilde\mu(i,h)-c\Big)-\frac{\gamma(1-\gamma)}{2}\|{\sigma(i)}^{\top}(I_n-\operatorname{diag}(h))\pi\|^2\Big]\varphi(t,i,h)
\nonumber\\
\hspace{-0.3cm}&&\hspace{-2.5cm}
+\sum_{j=1}^n\Big[(1-\pi^j)^{1-\gamma}\varphi(t,i,\widetilde{h}^j)-\varphi(t,i,h)\Big]\lambda_j(i,h)(1-h^j),\nonumber
\end{eqnarray}
and the $\mathbb{P}$-(local) martingale is defined as
\begin{eqnarray}
N^{\pi,c}_t
\hspace{-0.3cm}&:=&\hspace{-0.3cm}
\int_0^{t}\left(X^{\pi,c}_{v-}\right)^{1-\gamma}\varphi(v,I_{v-},H_{v-})\pi_v^{\top}{\big(I_n-\operatorname{diag}(H_{v-})\big)}{\sigma(I_{v-})dW_v}
\nonumber\\
\hspace{-0.3cm}&&\hspace{-0.3cm}
+\int_0^{t}\sum_{j=1}^n\frac{1}{1-\gamma}\left(X^{\pi,c}_{v-}\right)^{1-\gamma}\Big[(1-\pi_v^j)^{1-\gamma}\varphi(v,I_{v-},\widetilde{H}^j_{v-})-\varphi(v,I_{v-},{H}_{v-})\Big]dM^j_{v}
\nonumber\\
\hspace{-0.3cm}&&\hspace{-0.3cm}
+\int_0^{t}\sum_{\ell\neq I_{v-}}\frac{1}{1-\gamma}\left(X^{\pi,c}_{v-}\right)^{1-\gamma}\Big[\varphi(v,\ell,{H}_{v-})-\varphi(v,I_{v-},{H}_{v-})\Big]d\widetilde{K}^{I_{v-},\ell}_v,\nonumber
\end{eqnarray}
with
$\widetilde{K}^{i,j}_t=K^{i,j}_t-\int_0^tq_{ij}{\mathbf{1}_{\{I_{s-}=i\}}}ds,\, \text{for all }i,j\in\mathcal S_I\text{ and }i\neq j.$
Since $\varphi$ is the positive classical solution of \eqref{HJB}, the pointwise minimizing properties of $\pi^*$ and $c^*$ in \eqref{opti.strategy.inv} and \eqref{opti.strategy.cons} imply, for every $(\pi,c)\in\mathcal{U}(h)\times \mathbb{R}_+$,
\begin{eqnarray}
\label{4.34}
\mathcal{L}(\pi,c;t,i,h)+(c)^{\eta(1-\gamma)}\geq\mathcal{L}(\pi^*(t,i,h),c^*(t,i,h);t,i,h)+(c^*(t,i,h))^{\eta(1-\gamma)}=0.
\end{eqnarray}
Taking conditional expectation on both sides of \eqref{verify.equa} and using \eqref{4.34} yields that
\begin{eqnarray}\label{verify.condi2}
\hspace{-0.3cm}&&\hspace{-0.3cm}
\mathbb{E}_{t,x,i,h}\bigg[\frac{1}{1-\gamma}\left(X^{\pi,c}_{T\wedge\tau_\nu^t}\right)^{1-\gamma}\varphi(T\wedge\tau_\nu^t,I_{T\wedge\tau_\nu^t},H_{T\wedge\tau_\nu^t})\bigg]
\nonumber\\
\hspace{-0.3cm}&=&\hspace{-0.3cm}
\frac{x^{1-\gamma}}{1-\gamma}\varphi(t,i,h)
+\mathbb{E}_{t,x,i,h}\bigg[\int_t^{T\wedge\tau_\nu^t}\frac{1}{1-\gamma}\left(X^{\pi,c}_{s}\right)^{1-\gamma}\mathcal{L}(\pi_s,c_s;s,I_{s-},H_{s-})ds+N^{\pi,c}_{T\wedge\tau_\nu^t}-N^{\pi,c}_t\bigg]
\nonumber\\
\hspace{-0.3cm}&\leq&\hspace{-0.3cm}
\frac{x^{1-\gamma}}{1-\gamma}\varphi(t,i,h)
-\mathbb{E}_{t,x,i,h}\bigg[\int_t^{T\wedge\tau_\nu^t}\frac{1}{1-\gamma}\left(X^{\pi,c}_{s}\right)^{1-\gamma}(c_s)^{\eta(1-\gamma)}ds\bigg]
\nonumber\\
\hspace{-0.3cm}&=&\hspace{-0.3cm}
\frac{x^{1-\gamma}}{1-\gamma}\varphi(t,i,h)
-\mathbb{E}_{t,x,i,h}\bigg[\int_t^{T\wedge\tau_\nu^t}U_\eta(c_sX^{\pi,c}_{s},X^{\pi,c}_{s})ds\bigg].
\end{eqnarray}
It is easy to rewrite \eqref{verify.condi2} as follows:
\begin{eqnarray}
\label{4.36}
\hspace{-0.3cm}&&\hspace{-0.3cm}
\mathbb{E}_{t,x,i,h}\bigg[
\frac{1}{1-\gamma}\left(X^{\pi,c}_{T\wedge\tau_\nu^t}\right)^{1-\gamma}\varphi(T\wedge\tau_\nu^t,I_{T\wedge\tau_\nu^t},H_{T\wedge\tau_\nu^t})
\nonumber\\
\hspace{-0.3cm}&&\hspace{4.5cm}
+\int_t^{T\wedge\tau_\nu^t}U_\eta(c_sX^{\pi,c}_{s},X^{\pi,c}_{s})ds\bigg]
\leq
\frac{x^{1-\gamma}}{1-\gamma}
\varphi(t,i,h).
\end{eqnarray}
By $(\pi,c)\in\mathcal{A}$ and the boundedness of $\varphi$, the family
\begin{eqnarray}
\bigg(\frac{1}{1-\gamma}\left(X^{\pi,c}_{T\wedge\tau_\nu^t}\right)^{1-\gamma}\varphi(T\wedge\tau_\nu^t,I_{T\wedge\tau_\nu^t},H_{T\wedge\tau_\nu^t})\bigg)_{\nu\geq1}\nonumber
\end{eqnarray}
is dominated by an integrable random variable and hence is uniformly integrable. This, together with the fact that
$$\left(X^{\pi,c}_{T\wedge\tau_\nu^t}\right)^{1-\gamma}\varphi(T\wedge\tau_\nu^t,I_{T\wedge\tau_\nu^t},H_{T\wedge\tau_\nu^t})
\stackrel{\text{a.s.}}\longrightarrow
(X^{\pi,c}_T)^{1-\gamma}\varphi(T,I_T,H_T),\quad \text{as }\,\, \nu\rightarrow\infty,$$
yields that, as $\nu\rightarrow\infty$,
\begin{eqnarray}\label{L1}
\frac{1}{1-\gamma}\left(X^{\pi,c}_{T\wedge\tau_\nu^t}\right)^{1-\gamma}\varphi(T\wedge\tau_\nu^t,I_{T\wedge\tau_\nu^t},H_{T\wedge\tau_\nu^t})\stackrel{L^1}\longrightarrow \frac{1}{1-\gamma}(X^{\pi,c}_T)^{1-\gamma}\varphi(T,I_T,H_T).
\end{eqnarray}
It follows from \eqref{L1} that, as $\nu\to\infty$,
\begin{eqnarray}\label{L1.1}
\hspace{-0.3cm}&&\hspace{-0.3cm}
\mathbb{E}_{t,x,i,h}\bigg[\frac{1}{1-\gamma}\left(X^{\pi,c}_{T\wedge\tau_\nu^t}\right)^{1-\gamma}\varphi(T\wedge\tau_\nu^t,I_{T\wedge\tau_\nu^t},H_{T\wedge\tau_\nu^t})\bigg]
\nonumber\\
\hspace{-0.3cm}&\rightarrow&\hspace{-0.3cm}
\mathbb{E}_{t,x,i,h}\bigg[\frac{1}{1-\gamma}(X^{\pi,c}_T)^{1-\gamma}\varphi(T,I_T,H_T)\bigg].
\end{eqnarray}
Since $-U_\eta\geq0$ and $T\wedge\tau_\nu^t\uparrow T$, the monotone convergence theorem applied to $-U_\eta$ gives
\begin{eqnarray}\label{lim.u}
\lim_{\nu\to\infty}\mathbb{E}_{t,x,i,h}\bigg[\int_t^{T\wedge \tau_\nu^t}U_\eta(c_sX^{\pi,c}_{s},X^{\pi,c}_{s})ds\bigg]=\mathbb{E}_{t,x,i,h}\bigg[\int_t^{T}U_\eta(c_sX^{\pi,c}_s,X^{\pi,c}_s)ds\bigg].
\end{eqnarray}
Then, in lieu of \eqref{4.36}, \eqref{L1.1} and \eqref{lim.u}, we have
\begin{eqnarray}
\hspace{-0.3cm}&&\hspace{-0.3cm}
\mathbb{E}_{t,x,i,h}\bigg[
u\left(X^{\pi,c}_{T}\right)+\int_t^{T}U_\eta(c_sX^{\pi,c}_{s},X^{\pi,c}_{s})ds\bigg]
\nonumber\\
\hspace{-0.3cm}&=&\hspace{-0.3cm}
\mathbb{E}_{t,x,i,h}\bigg[
\frac{1}{1-\gamma}\left(X^{\pi,c}_{T}\right)^{1-\gamma}\varphi(T,I_{T},H_{T})+\int_t^{T}U_\eta(c_sX^{\pi,c}_{s},X^{\pi,c}_{s})ds\bigg]
\nonumber\\
\hspace{-0.3cm}&=&\hspace{-0.3cm}
\lim_{\nu\to\infty}
\mathbb{E}_{t,x,i,h}\bigg[
\frac{1}{1-\gamma}\left(X^{\pi,c}_{T\wedge\tau_\nu^t}\right)^{1-\gamma}\varphi(T\wedge\tau_\nu^t,I_{T\wedge\tau_\nu^t},H_{T\wedge\tau_\nu^t})
\nonumber\\
\hspace{-0.3cm}&&\hspace{1cm}\quad \quad \,\,
+\int_t^{T\wedge\tau_\nu^t}U_\eta(c_sX^{\pi,c}_{s},X^{\pi,c}_{s})ds\bigg]
\nonumber\\
\hspace{-0.3cm}&\leq&\hspace{-0.3cm}
\frac{x^{1-\gamma}}{1-\gamma}
\varphi(t,i,h).\nonumber
\end{eqnarray}
Above, in the first equality, one has used the terminal condition $\varphi(T,i,h)=1$ for all $(i,h)\in\mathcal S_I\times\mathcal S_H$. The validity of conclusion (\romannumeral1) is verified.

We next turn to the proof of conclusion (\romannumeral2). We first prove that the portfolio-consumption strategy $(\pi^*,c^*)$  defined by \eqref{opti.strategy.inv} and \eqref{opti.strategy.cons} is admissible. Fix an initial state $(t,x,i,h)\in[0,T]\times\mathbb R_+\times\mathcal S_I\times\mathcal S_H$. The associated feedback processes are
\[\pi_s^*:=\pi^*(s,I_{s-},H_{s-}),\qquad
c_s^*:=c^*(s,I_{s-},H_{s-}),\qquad s\in(t,T],\]
with initial values $\pi_t^*=\pi^*(t,i,h)$ and $c_t^*=c^*(t,i,h)$. 	Furthermore, $0<C_{\min}\leq\varphi\leq C_{\max}$ gives that $c^*$ is bounded.
Under $(\pi^*,c^*)$ and the initial condition $X_t^{\pi^*,c^*}=x$, the investor's wealth process satisfies, for every $u\in[t,T]$,
\begin{eqnarray}
\label{4.40}
\hspace{-0.3cm}&&\hspace{-1.5cm}|X^{\pi^*,c^*}_u|^{1-\gamma}
=
x^{1-\gamma}\exp\Bigg((1-\gamma)\int_t^u\big[{r(I_s)}+\pi_s^{*\top}
(\mu(I_s)-r(I_s)e_n)
-c^*_s\big]ds
\nonumber\\
\hspace{-0.3cm}&&\hspace{-0.3cm}
+(1-\gamma)\int_t^u\pi_s^{*\top}{\sigma(I_s)dW_s}
-(1-\gamma)\int_t^u\frac{1}{2}\|{\sigma(I_s)}^{\top}\pi^*_s\|^2ds
\nonumber\\
\hspace{-0.3cm}&&\hspace{-0.3cm}
+\sum_{j=1}^{n}(1-\gamma)\int_t^{u}\log(1-\pi^{j,*}_s)dH_s^j
+\sum_{j=1}^{n}(1-\gamma)\int_t^{u}(1-H_{s-}^j)\lambda_j(I_s,H_s)\pi^{j,*}_sds
\Bigg)
\nonumber\\
\hspace{-0.3cm}&=&\hspace{-0.3cm}
x^{1-\gamma}\exp\Bigg((1-\gamma)\int_t^u\big[{r(I_s)+\pi_s^{*\top}}
(\mu(I_s)-r(I_s)e_n)
-{\Big(\frac{{\varphi}(s,I_s,H_s)}{\eta}\Big)^{\frac{1}{\eta(1-\gamma)-1}}}\big]ds
\nonumber\\
\hspace{-0.3cm}&&\hspace{-0.3cm}
+(1-\gamma)\int_t^u\pi_s^{*\top}{\sigma(I_s)dW_s}
{-(1-\gamma)\int_t^u\frac{1}{2}\|{\sigma(I_s)}^{\top}\pi^*_s\|^2ds}
\nonumber\\
\hspace{-0.3cm}&&\hspace{-0.3cm}
+\sum_{j=1}^{n}(1-\gamma)\int_t^{u}\log(1-\pi^{j,*}_s)dH_s^j
+\sum_{j=1}^{n}(1-\gamma)\int_t^{u}(1-H_{s-}^j)\lambda_j(I_s,H_s)\pi^{j,*}_sds
\Bigg)
\nonumber\\
\hspace{-0.3cm}&=&\hspace{-0.3cm}
x^{1-\gamma}\exp\Bigg(\int_t^u(1-\gamma)\Big[{r(I_s)}+\pi^{*\top}_s(\mu(I_s)-r(I_s)e_n)
-{\Big(\frac{{\varphi}(s,I_s,H_s)}{\eta}\Big)^{\frac{1}{\eta(1-\gamma)-1}}}\big]ds
\nonumber\\
\hspace{-0.3cm}&&\hspace{-0.3cm}
+\sum_{j=1}^{n}(1-\gamma)\int_t^{u}\log(1-\pi^{j,*}_s)dH_s^j
+\sum_{j=1}^{n}(1-\gamma)\int_t^{u}(1-H_{s-}^j)\lambda_j(I_s,H_s)\pi^{j,*}_sds
\nonumber\\
\hspace{-0.3cm}&&\hspace{-0.3cm}
-\int_t^u\frac{\gamma(1-\gamma)}{2}\|{\sigma(I_s)}^{\top}\pi^*_s\|^2ds\Bigg)
\nonumber\\
\hspace{-0.3cm}&&\hspace{-0.3cm}
\times\exp\Bigg((1-\gamma)\int_t^u\pi^{*\top}_s{\sigma(I_s)dW_s}
-\frac{(1-\gamma)^2}{2}\int_t^u\|{\sigma(I_s)}^{\top}\pi^*_s\|^2ds
\Bigg).
\end{eqnarray}
It follows from Lemma \ref{lem4.5} that
$$\max_{1\leq j\leq n}|\log(1-{\pi}^{j,*})|\leq\max\left\{|\log C_{-}|,\log(1+C_+)\right\}:=\Bar{C},$$
with $C_+$ and $C_{-}$ being given by \eqref{C3} and \eqref{C.hat}, respectively.
Hence, for every $u\in[t,T]$, we deduce from \eqref{4.40} that
\begin{eqnarray}
\label{4.69}
\hspace{14mm}|X^{\pi^*,c^*}_u|^{1-\gamma}\leq C_T x^{1-\gamma}\exp\Bigg((1-\gamma)\int_t^u\pi^{*\top}_s{\sigma(I_s)dW_s}
-\frac{(1-\gamma)^2}{2}\int_t^u\|{\sigma(I_s)}^{\top}\pi^*_s\|^2ds
\Bigg).\hspace{-14mm}
\end{eqnarray}
Here, $C_{T}>0$ is a finite constant given as
\begin{eqnarray}\label{4.41}
C_{T}\hspace{-0.3cm}&:=&\hspace{-0.3cm}\exp\Bigg\{T\bigg[\left|1-\gamma\right|\Big[C_{r}+C_+C_{\bar{\mu}}
+\Big(\frac{C_{\min}}{\eta}\Big)^{\frac{1}{\eta(1-\gamma)-1}}\Big]
\nonumber\\
\hspace{-0.3cm}&&\hspace{0.7cm}
+nC_+|1-\gamma|\lambda_{\max}
+C_+^2\left|\frac{\gamma(1-\gamma)}{2}\right|\delta_{\max}^2\bigg]+n|1-\gamma|\Bar{C}\Bigg\},
\end{eqnarray}
with $C_r:=\max_{1\leq i\in\mathcal S_I}r(i)$ {and $C_{\bar\mu}:=\max_{i\in\mathcal S_I}\|\mu(i)-r(i)e_n\|$}.
By Lemma~\ref{lem4.5} and the definition of $\delta_{\max}$, we have
\[
\|\sigma(I_s)^\top\pi_s^*\|\leq\delta_{\max}C_+,
\qquad s\in[t,T].
\]
Hence, with Novikov's condition verified, we have
\begin{align*}
&\mathbb E_{t,x,i,h}\Bigg[\exp\Bigg(
2(1-\gamma)\int_t^T\pi_s^{*\top}\sigma(I_s)dW_s
-(1-\gamma)^2\int_t^T\|\sigma(I_s)^\top\pi_s^*\|^2ds
\Bigg)\Bigg]\\
&=\mathbb E_{t,x,i,h}\Bigg[\exp\Bigg(
2(1-\gamma)\int_t^T\pi_s^{*\top}\sigma(I_s)dW_s
-2(1-\gamma)^2\int_t^T\|\sigma(I_s)^\top\pi_s^*\|^2ds
\Bigg)\\
&\qquad\qquad\times\exp\left((1-\gamma)^2\int_t^T\|\sigma(I_s)^\top\pi_s^*\|^2ds\right)\Bigg]
\\
&\qquad\leq
\exp\big((1-\gamma)^2\delta_{\max}^2C_+^2(T-t)\big)<\infty.
\end{align*}
The preceding calculation bounds the terminal second moment of the stochastic exponential
\[\exp\Bigg((1-\gamma)\int_t^u\pi_s^{*\top}\sigma(I_s)dW_s-\frac{(1-\gamma)^2}{2}\int_t^u\|\sigma(I_s)^\top\pi_s^*\|^2ds\Bigg),\qquad u\in[t,T].\]
This process is a martingale by Novikov's condition and is squareintegrable by the above estimate.
Applying Doob's $L^2$ maximal inequality to this process gives
\begin{align*}
&\mathbb E_{t,x,i,h}\Bigg[\sup_{u\in[t,T]}\exp\Bigg(
2(1-\gamma)\int_t^u\pi_s^{*\top}\sigma(I_s)dW_s
-(1-\gamma)^2\int_t^u\|\sigma(I_s)^\top\pi_s^*\|^2ds
\Bigg)\Bigg]\\
&\qquad\leq
4\exp\big((1-\gamma)^2\delta_{\max}^2C_+^2(T-t)\big).
\end{align*}
Combining this estimate with \eqref{4.69} and the Cauchy--Schwarz inequality yields
\begin{align*}
\mathbb E_{t,x,i,h}\left[\sup_{s\in[t,T]}(X_s^{\pi^*,c^*})^{1-\gamma}\right]
&\leq \left(\mathbb E_{t,x,i,h}\left[\sup_{s\in[t,T]}(X_s^{\pi^*,c^*})^{2(1-\gamma)}\right]\right)^{\frac{1}{2}}\\
&\leq
2C_Tx^{1-\gamma}
\exp\left(\frac{(1-\gamma)^2\delta_{\max}^2C_+^2(T-t)}{2}\right)<\infty.
\end{align*}
Moreover, $c^*$ is bounded and strictly positive, while Lemma~\ref{lem4.5} gives $\|\pi^*\|\leq C_+$ and $1-\pi^{j,*}\geq C_->0$ for every surviving asset. For a defaulted asset, $\pi^{j,*}=0$. 
Consequently, the wealth equation under $(\pi^*,c^*)$ admits a unique strictly positive strong solution, and the preceding estimate proves that $(\pi^*,c^*)\in\mathcal A$.

To prove the identity in conclusion (\romannumeral2), it follows from \eqref{verify.condi2} that
\begin{eqnarray}\label{verify.condi4}
\hspace{-0.3cm}&&\hspace{-0.3cm}
\mathbb{E}_{t,x,i,h}\bigg[\frac{1}{1-\gamma}\left(X^{\pi^*,c^*}_{T\wedge\tau_\nu^t}\right)^{1-\gamma}\varphi(T\wedge\tau_\nu^t,I_{T\wedge\tau_\nu^t},H_{T\wedge\tau_\nu^t})\bigg]
\nonumber\\
\hspace{-0.3cm}&=&\hspace{-0.3cm}
\frac{x^{1-\gamma}}{1-\gamma}\varphi(t,i,h)
+\mathbb{E}_{t,x,i,h}\bigg[\int_t^{T\wedge\tau_\nu^t}\frac{1}{1-\gamma}\left(X^{\pi^*,c^*}_{s}\right)^{1-\gamma}\mathcal{L}(\pi^*,c^*;s,I_s,H_s)ds\bigg].
\end{eqnarray}
By virtue of \eqref{HJB} and the minimizing properties in \eqref{opti.strategy.inv} and \eqref{opti.strategy.cons}, we have, a.s.
\begin{eqnarray}\label{A+c=0}
\mathcal{L}(\pi^*_s,c^*_s;s,I_s,H_s)+(c_s^*)^{\eta(1-\gamma)}=0,\quad s\in[t,T].
\end{eqnarray}
Thus, plugging \eqref{A+c=0} into \eqref{verify.condi4}, and then using \eqref{L1}--\eqref{lim.u}, one derives that
\begin{eqnarray}
J(\pi^*,c^*;t,x,i,h)
\hspace{-0.3cm}&=&\hspace{-0.3cm}
\mathbb{E}_{t,x,i,h}\bigg[\int_t^TU_\eta(c^*_sX^{\pi^*,c^*}_s,X^{\pi^*,c^*}_s)ds+u(X^{\pi^*,c^*}_T)\bigg]
\nonumber\\
\hspace{-0.3cm}&=&\hspace{-0.3cm}
\lim_{\nu\to\infty}\mathbb{E}_{t,x,i,h}\bigg[\frac{1}{1-\gamma}\left(X^{\pi^*,c^*}_{T\wedge\tau_\nu^t}\right)^{1-\gamma}\varphi(T\wedge\tau_\nu^t,I_{T\wedge\tau_\nu^t},H_{T\wedge\tau_\nu^t})
\nonumber\\
\hspace{-0.3cm}&&\hspace{4.5cm}
+\int_t^{T\wedge\tau_\nu^t}U_\eta(c^*_sX^{\pi^*,c^*}_{s},X^{\pi^*,c^*}_{s})ds\bigg]
\nonumber\\
\hspace{-0.3cm}&=&\hspace{-0.3cm}
\frac{x^{1-\gamma}}{1-\gamma}\varphi(t,i,h)
.\nonumber
\end{eqnarray}
Since $(\pi^*,c^*)\in\mathcal A$, conclusion (\romannumeral1) and the definition of $V$ now imply
\[
V(t,x,i,h)=J(\pi^*,c^*;t,x,i,h)
=\frac{x^{1-\gamma}}{1-\gamma}\varphi(t,i,h).
\]
This completes the proof of conclusion (\romannumeral2).
\end{proof}

\section{Numerical Analysis}\label{sec:numerical}
This section examines how default risk, diffusion volatility, regime switching, and preference parameters affect the optimal investment and consumption decisions. We conduct a sensitivity analysis of the optimal controls with respect to these parameters, and four matched comparisons: with versus without contagion, favorable versus adverse regimes, the consumption-inclusive objective versus its terminal-wealth-only counterpart, and decisions immediately before versus after a default.
We consider a market consisting of two defaultable risky assets and one riskless bond, i.e., $n=2$. In this market model, the risky-asset prices in \eqref{asset.price} specialize to
\begin{equation}\label{numerical.stock.dynamics}
\left\{
\begin{aligned}
S_t^1&=(1-H_t^1)\widehat S_t^1,&
\frac{d\widehat S_t^1}{\widehat S_t^1}
&=\bigl(\mu_1(I_t)+\lambda_1(I_t,H_t)\bigr)dt+\sum_{k=1}^2\sigma_{1k}(I_t)dW_t^k,\\
S_t^2&=(1-H_t^2)\widehat S_t^2,&
\frac{d\widehat S_t^2}{\widehat S_t^2}
&=\bigl(\mu_2(I_t)+\lambda_2(I_t,H_t)\bigr)dt+\sum_{k=1}^2\sigma_{2k}(I_t)dW_t^k.
\end{aligned}
\right.
\end{equation}
Here $W=(W^1,W^2)^\top$ is a two-dimensional Brownian motion (i.e., $d=2$), $H=(H^1,H^2)$ takes values in $\{0,1\}^2$, and $I$ is a continuous-time Markov chain with state space $\{1,2\}$ (i.e., $m=2$).  $\lambda_j(i,h)$ is the default intensity of asset $j$ when the market is in regime $i$ and the current default state is $h$, provided that asset $j$ is still alive, that is, $h^j=0$.  In particular, we use the notation $\lambda^h_j:= (\lambda_j(1, h), \lambda_j(2, h))$ to represent the vector of default intensities of the $j$-th stock at the default state $h$.  The riskless bond satisfies $dB_t=r(I_t)B_tdt$. For the portfolio fractions
$\pi_t=(\pi_t^1,\pi_t^2)^\top$ and the consumption rate $c_t$, the corresponding wealth dynamics are
\[ \frac{dX_t^{\pi,c}}{X_{t-}^{\pi,c}} =\bigl[r(I_t)+\pi_t^\top(\mu(I_t)-r(I_t)e_2)-c_t\bigr]dt +\pi_t^\top\sigma(I_t)dW_t-\pi_t^\top dM_t. \]
Here $M=(M^1,M^2)^\top$ is the compensated default martingale defined in \eqref{default.martingale}.

Throughout
the section, we use the  benchmark parameters given in Subsection \ref{subsec4.1}.

\subsection{Benchmark parameters}\label{subsec4.1}
We set $T=1$. The market coefficients in regimes 1 and 2 are
\[
\begin{gathered}
\mu(1)=(1.4,0.8)^\top,\qquad \mu(2)=(1,0.55)^\top,\qquad
r(1)=0.1,\quad r(2)=0.06,\\
Q_0=\begin{pmatrix}-0.5&0.5\\1&-1\end{pmatrix},\qquad
\sigma(1)=\begin{pmatrix}0.7&0\\0&1\end{pmatrix},\qquad
\sigma(2)=\begin{pmatrix}1&0\\0&1.5\end{pmatrix}.
\end{gathered}
\]
We take $\gamma=2.5$ and $\eta=0.5$ unless stated otherwise. The benchmark values of $\lambda^h_j$ are
\begin{equation}\label{numerical.lambda}
\begin{aligned}
\lambda_1^{(0,0)}&=(0.5,0.75),&
\lambda_2^{(0,0)}&=(0.75,1.1),\\
\lambda_1^{(0,1)}&=(0.7,1.0),&
\lambda_2^{(1,0)}&=(0.9,1.3).
\end{aligned}
\end{equation}
The intensities of an already defaulted asset are immaterial. The inequalities
\[\lambda_1^{(0,1)}-\lambda_1^{(0,0)}=(0.2,0.25)>0,
\qquad\lambda_2^{(1,0)}-\lambda_2^{(0,0)}=(0.15,0.2)>0\]
show that, in either regime, the default of one asset increases the default intensity of the surviving asset. They therefore represent the default-contagion effect in the benchmark calibration. For every parameter value considered below, we solve the recursive HJB systems successively for $h=(1,1)$, $(0,1)$, $(1,0)$, and $(0,0)$, and recover the feedback controls from \eqref{opti.strategy.cons} and \eqref{opti.strategy.inv}.

\subsection{Default risk and volatility after one default}
We first consider states in which only one risky asset remains alive. Figure~\ref{fig:numerical-intensity} varies the surviving asset's current default intensity. In the top panels, asset 1 survives in state $(0,1)$ and the market is in regime 2; in the bottom panels, asset 2 survives in state $(1,0)$ and the market is in regime 1.
\begin{figure}[H]
\centering
\includegraphics[width=0.94\textwidth]{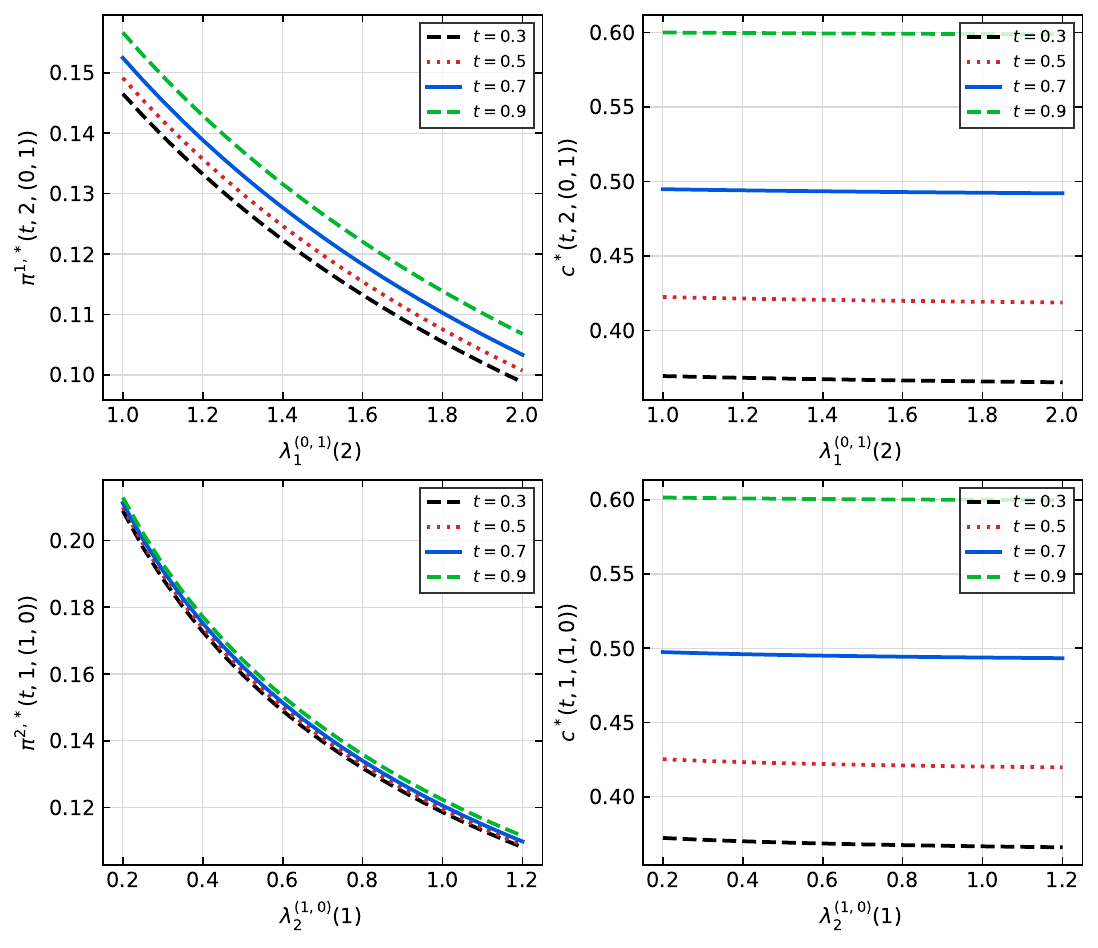}
\caption{Dependence of the optimal portfolio fraction and consumption rate on the current default intensity after one default. }
\label{fig:numerical-intensity}
\end{figure}

The left panels of Figure~\ref{fig:numerical-intensity} show that the risky-asset allocation decreases monotonically with the asset's own default intensity. This monotonicity reflects two opposing effects of the intensity. A higher intensity enlarges the jump-to-default loss exposure of the position, which discourages investment in the asset; at the same time, through the drift specification in \eqref{asset.price}, it raises the premium earned while the asset survives, which encourages investment. The observed decline shows that the loss-exposure effect dominates the premium effect.

The right panels show that the consumption rate reacts much more mildly: over the same intensity ranges, $c_t^*$ declines from $0.3695$ to $0.3653$ in the top-right panel and from $0.3723$ to $0.3660$ in the bottom-right panel. This contrast with the portfolio response reflects the feedback structure of the optimal controls. The default intensity directly affects the portfolio choice through its expression in \eqref{opti.strategy.inv}, whereas it affects consumption only indirectly through  \(\varphi\) in \eqref{opti.strategy.cons}.

Figure~\ref{fig:numerical-volatility} gives the analogous experiment for diffusion volatility. We vary $\sigma_{11}(1)$ in state $(0,1)$ and $\sigma_{22}(1)$ in state $(1,0)$, keeping the remaining entries fixed. The left panels of Figure~\ref{fig:numerical-volatility} show that the risky-asset allocation likewise decreases monotonically with the asset's own diffusion volatility. This is exactly consistent with the intuition. The right panels show that the consumption rate moves in the same direction but much more mildly. The same feedback structure as in the default-intensity experiment accounts for the weaker consumption response.

\begin{figure}[H]
\centering
\includegraphics[width=0.94\textwidth]{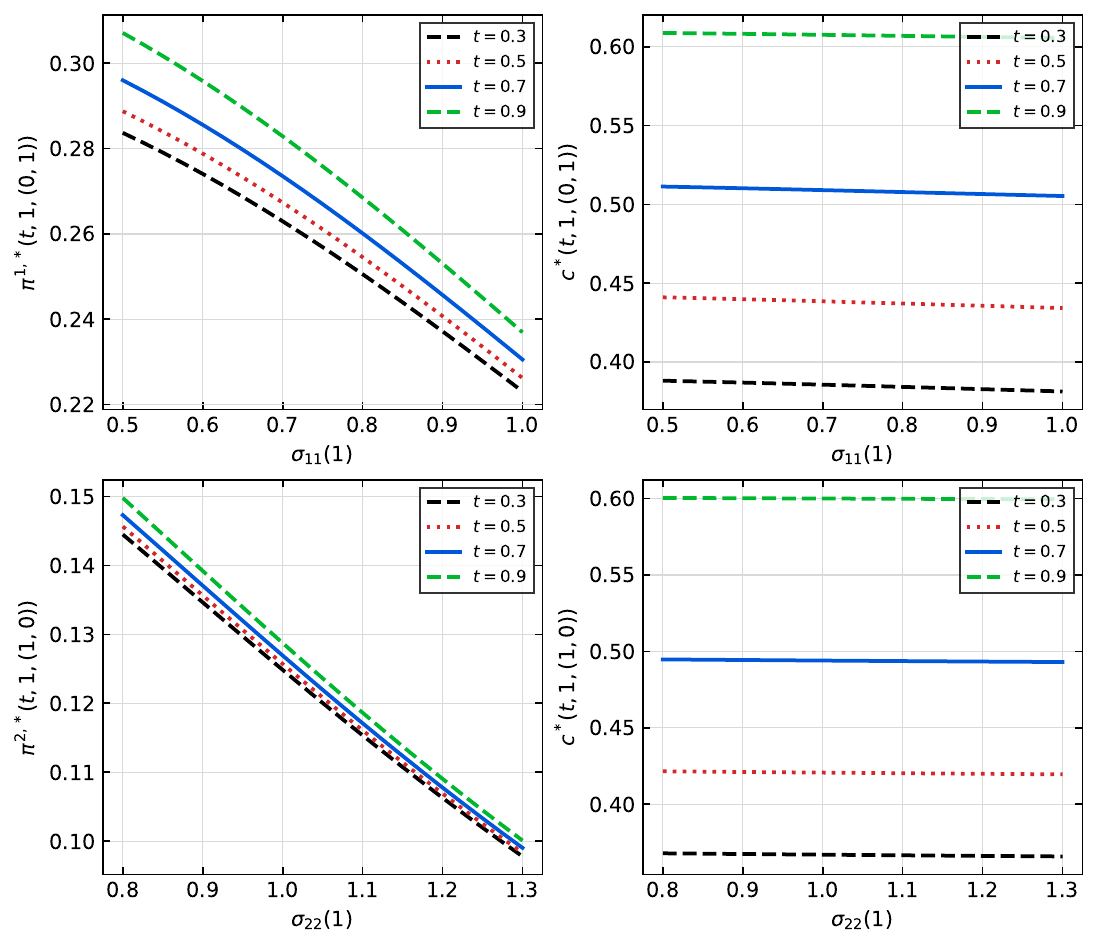}
\caption{Dependence of the optimal portfolio fraction and consumption rate in regime 1 on the volatility of the surviving asset.}
\label{fig:numerical-volatility}
\end{figure}

\subsection{Direct and indirect effects before default}\label{subsec4.3}
We next place both assets in the alive state $h=(0,0)$. Figures~\ref{fig:numerical-asset2-intensity} and \ref{fig:numerical-asset1-intensity} show how the optimal portfolio fractions of both assets vary with $\lambda_2^{(0,0)}(1)$ in regime 1 and $\lambda_1^{(0,0)}(2)$ in regime 2, respectively; in each figure, the direct response of the asset whose intensity rises is shown alongside the indirect response of the other asset.
It can be seen that an increase in an asset's default intensity substantially reduces the optimal fraction invested in that asset, while slightly increasing the fraction invested in the other asset. Overall, the total fraction of wealth invested in the two risky assets, $\pi^{1,*}+\pi^{2,*}$, decreases. The decline in the former indicates that the increased jump-to-default exposure dominates the accompanying increase in the default-risk premium. The response of the other asset reflects two opposing effects. On one hand, in Figure~\ref{fig:numerical-asset2-intensity}, an increase in \(\lambda_2^{(0,0)}(1)\) makes the default of asset 2 more likely. If asset 2 defaults, the default state changes from \((0,0)\) to \((0,1)\), and contagion raises the default intensity of the surviving asset 1 from \(\lambda_1^{(0,0)}(1)=0.5\) to \(\lambda_1^{(0,1)}(1)=0.7\). The investor therefore anticipates a greater likelihood that asset 1 will face higher default risk in the future, which exerts downward pressure on its current portfolio fraction. The same mechanism operates in Figure~\ref{fig:numerical-asset1-intensity} with the roles of the two
assets exchanged.  On the other hand, as the asset whose default intensity has increased becomes less attractive, the investor reduces the portfolio fraction invested in that asset and reallocates part of the wealth to the other asset. The portfolio fraction of the other asset increases slightly in both figures, indicating that the reallocation away from the asset whose current default intensity rises outweighs this adverse anticipated-contagion effect.

\begin{figure}[H]
\centering
\includegraphics[width=0.94\textwidth]{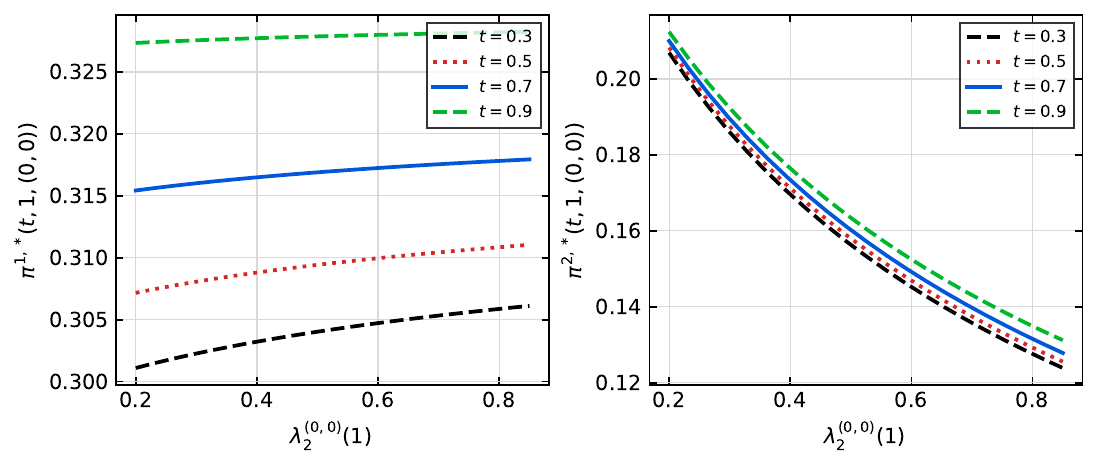}
\caption{Dependence of both portfolio fractions on
$\lambda_2^{(0,0)}(1)$ in regime 1.}
\label{fig:numerical-asset2-intensity}
\end{figure}

\begin{figure}[H]
\centering
\includegraphics[width=0.94\textwidth]{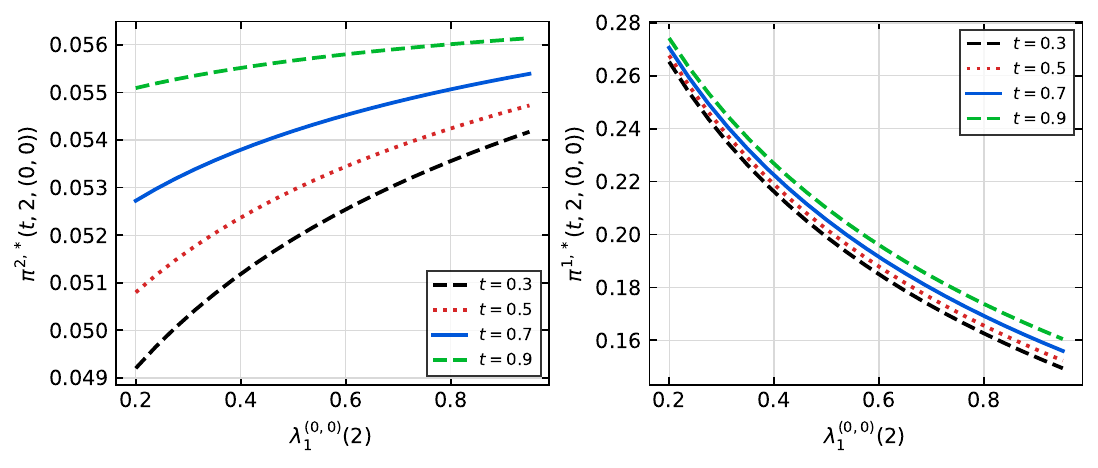}
\caption{Dependence of both portfolio fractions on
$\lambda_1^{(0,0)}(2)$ in regime 2. }
\label{fig:numerical-asset1-intensity}
\end{figure}

\subsection{Regime switching and preference parameters}\label{Subsec4.4}
To examine how the frequency of regime changes affects $\varphi$, we take $Q=Q_0$ as the benchmark generator and consider $sQ_0$ for $s\in\{1,2,3,4\}$. Multiplying the generator by $s$ increases all transition intensities proportionally, while leaving the stationary distribution unchanged. Figure~\ref{fig:numerical-regime-value} compares $\varphi(t,1,h)-\varphi(t,2,h)$ under these switching frequencies for each default state $h$.

\begin{figure}[H]
\centering
\includegraphics[width=0.94\textwidth]{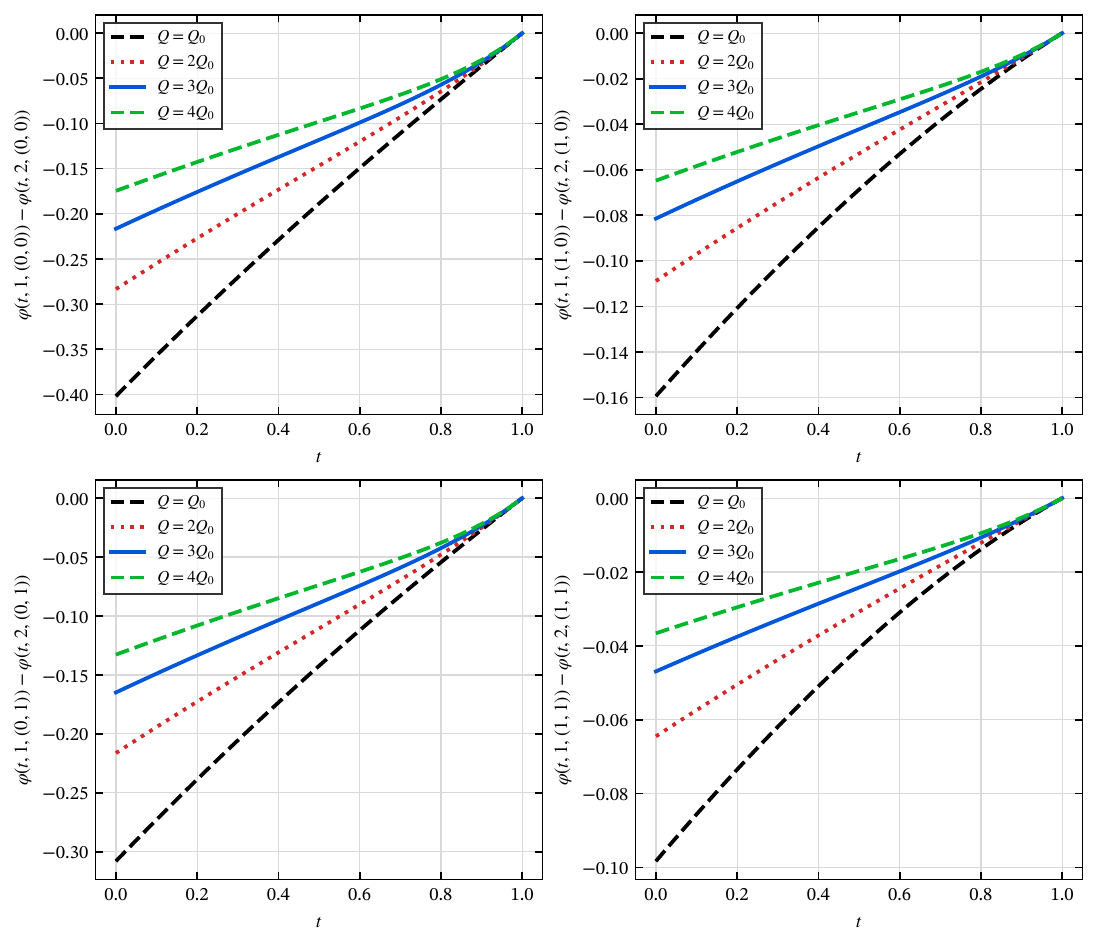}
\caption{Difference $\varphi(t,1,h)-\varphi(t,2,h)$ under
$Q=Q_0,2Q_0,3Q_0,$ and $4Q_0$, for $h=(0,0),(1,0),(0,1)$, and $(1,1)$.}
\label{fig:numerical-regime-value}
\end{figure}

Across all default states, the difference of $\varphi$  between regimes becomes smaller as the switching frequency increases. For example, at $t=0$, its absolute value decreases from $0.4019$ under $Q_0$ to $0.1746$ under $4Q_0$ in state $(0,0)$, and from $0.0983$ to $0.0366$ in state $(1,1)$. Faster switching makes the Markov chain less likely to remain in its initial regime for a long period. Consequently, the future regime distributions conditional on the two initial regimes become more similar, reducing the difference between $\varphi(t,1,h)$ and $\varphi(t,2,h)$. Finally, the difference vanishes at maturity because the terminal condition $\varphi(T,i,h)=1$ is independent of the regime.

Finally, Figure~\ref{fig:numerical-eta} varies the Cobb--Douglas weight $\eta\in[0.2,1]$ in regime 1 and state $(0,0)$. A larger $\eta$ places more weight on consumption expenditure relative to current wealth in the running utility. The optimal consumption rate responds much more strongly to $\eta$ than the portfolio fraction. At $t=0.3$, increasing $\eta$ from $0.2$ to $1$ raises $c_t^*$ from $0.1722$ to $0.6494$, whereas $\pi_t^{1,*}$ decreases only slightly, from $0.3061$ to $0.3042$. Economically, a larger $\eta$ places greater importance on current consumption expenditure relative to wealth retention, encouraging the investor to consume a larger proportion of wealth. Since this preference change does not alter the returns, volatilities, or default risks of the available assets, it has only a limited effect on credit allocation.

To examine risk aversion separately, Figure~\ref{fig:numerical-gamma} varies $\gamma$ over $[1.2,5]$ in regime 1 and state $(0,0)$. Figure~\ref{fig:numerical-gamma} shows that greater risk aversion substantially reduces credit exposure. At $t=0.3$, increasing $\gamma$ from $1.2$ to 5 lowers $\pi_t^{1,*}$ from $0.5771$ to $0.1584$, while $c_t^*$ rises from $0.3213$ to $0.4300$. Economically, a more risk-averse investor is less willing to expose wealth to diffusion and default risk and therefore holds a smaller position in asset 1. The reduction in risky investment is accompanied by a higher current consumption rate, indicating that a more risk-averse investor consumes a larger share of current wealth rather than carrying it forward.

\begin{figure}[H]
\centering
\includegraphics[width=0.94\textwidth]{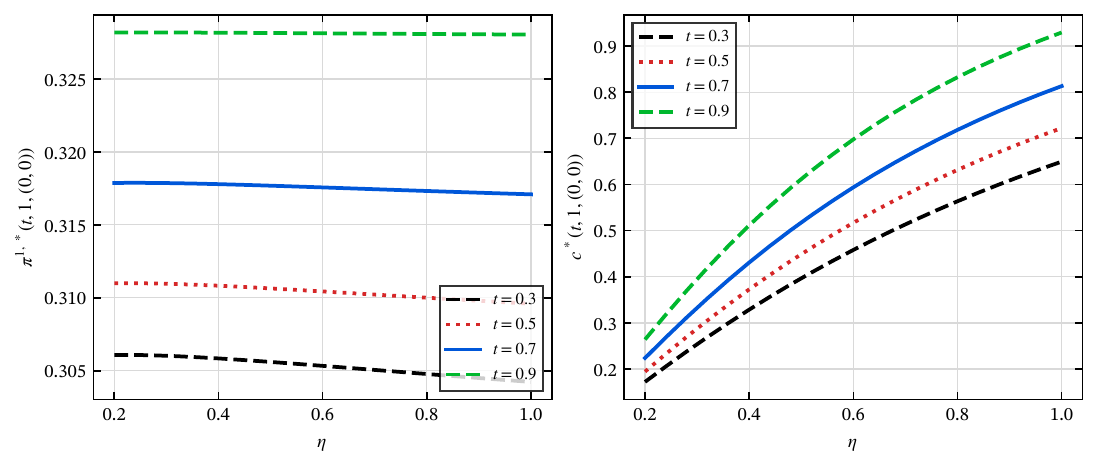}
\caption{Dependence of the optimal portfolio fraction $\pi^{1,*}$ and consumption rate on the Cobb--Douglas weight $\eta$ in regime 1 and state $(0,0)$.}
\label{fig:numerical-eta}
\end{figure}

\begin{figure}[H]
\centering
\includegraphics[width=0.94\textwidth]{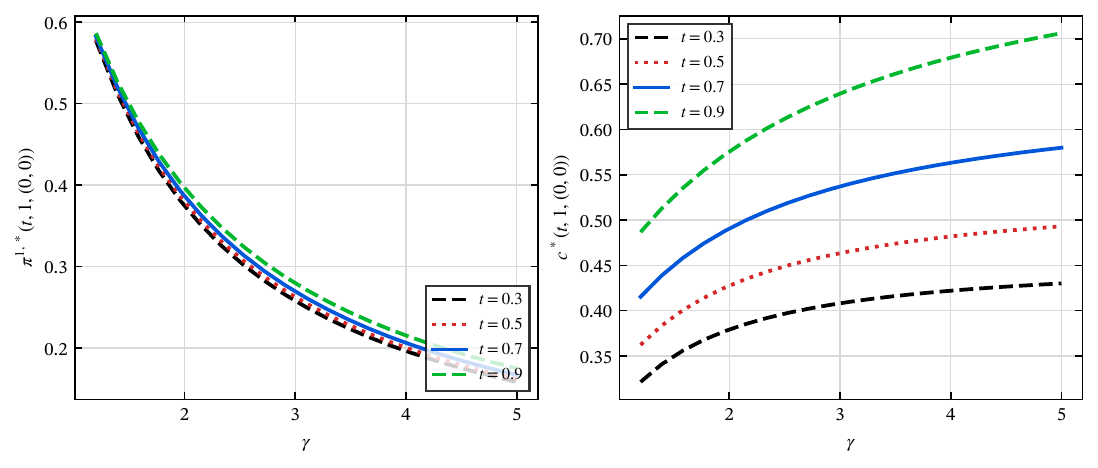}
\caption{Dependence of the optimal portfolio fraction $\pi^{1,*}$ and consumption rate on the risk-aversion parameter $\gamma$ in regime 1 and state $(0,0)$.}
\label{fig:numerical-gamma}
\end{figure}

\subsection{Contagion versus no contagion}\label{subsec:matched-contagion}

We compare the optimal investment and consumption strategies with and without default contagion, holding fixed the default intensities in the all-alive state. In the no-contagion benchmark, we set
\[\lambda_1^{(0,1)}=\lambda_1^{(0,0)}=(0.5,0.75),\qquad
\lambda_2^{(1,0)}=\lambda_2^{(0,0)}=(0.75,1.1).
\]
Thus, an individual default no longer increases the surviving asset's intensity. The contagion model uses \eqref{numerical.lambda}. 

Figure~\ref{fig:matched-contagion} compares the optimal controls in state $(0,0)$ with and without contagion. At $t=0.3$ in regime 1, introducing contagion moves $(\pi^{1,*},\pi^{2,*})$ from $(0.3050,0.1358)$ to $(0.3056,0.1315)$, so total risky investment falls from $0.4407$ to $0.4371$; the consumption rate falls from $0.3981$ to $0.3966$. The same directions of response hold in regime 2 at this time. Thus, anticipated contagion lowers aggregate credit exposure and the consumption rate.

\begin{figure}[H]
\centering
\includegraphics[width=0.98\textwidth]{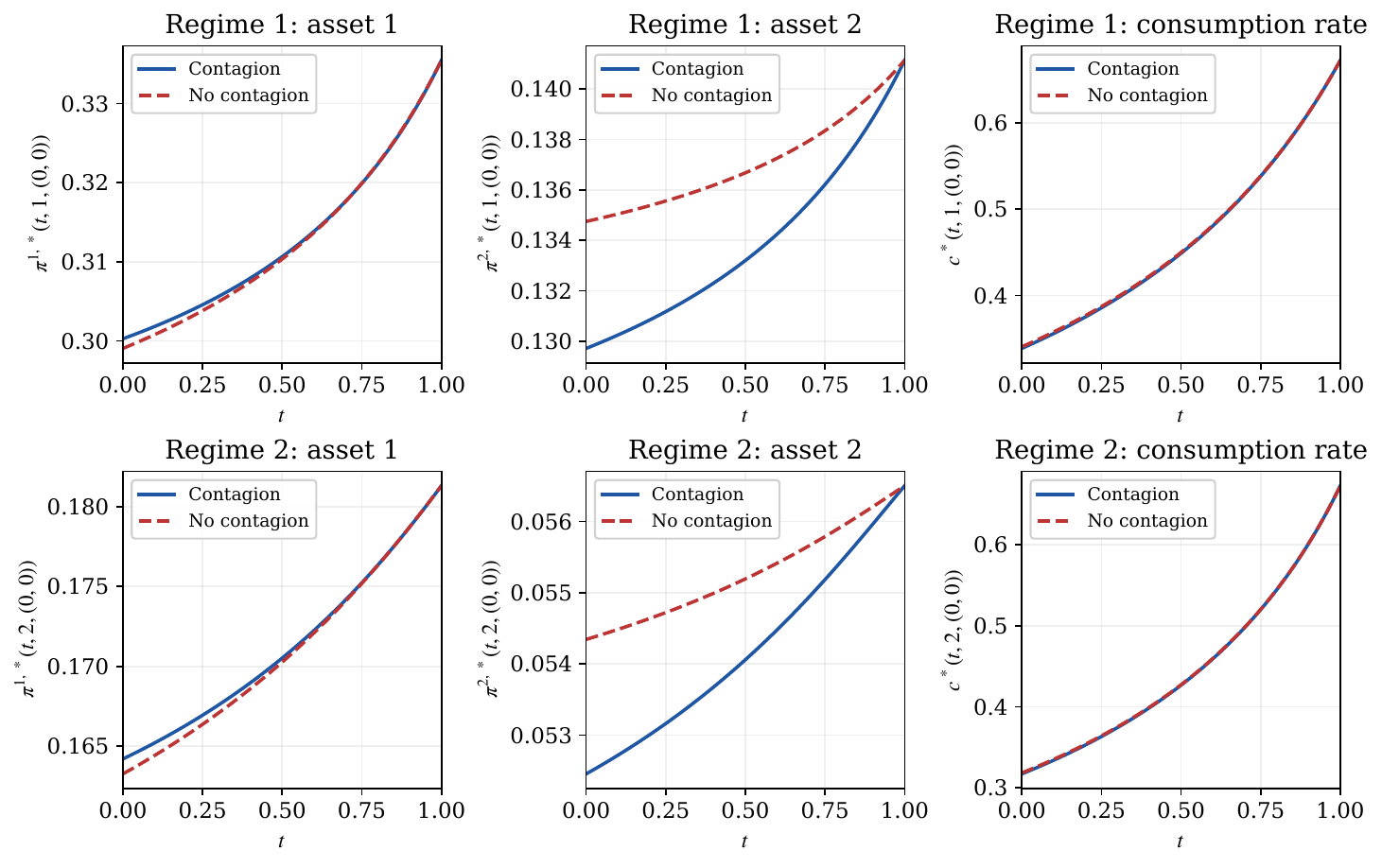}
\caption{Optimal portfolio fractions and consumption rates with and without contagion in state $(0,0)$.}
\label{fig:matched-contagion}
\end{figure}

\subsection{Comparison across regimes}\label{subsec:matched-regimes}
Regime 1 has higher expected asset returns and a higher riskless rate, together with lower volatilities and default intensities, than regime 2. We therefore interpret the two regimes as relatively favorable and adverse, respectively. Figure~\ref{fig:matched-regimes} compares the feedback controls of the two regimes at the default state $(0,0)$ in the contagion model. At $t=0.3$, the portfolio fractions are $(0.3056,0.1315)$ in regime 1 and $(0.1676,0.0533)$ in regime 2, and the consumption rate declines from $0.3966$ to $0.3742$. Thus, moving from the favorable to the adverse regime lowers both credit exposure and the consumption-to-wealth ratio: at a given wealth level, the investor consumes less and allocates a larger fraction of wealth to the riskless asset.

\begin{figure}[H]
\centering
\includegraphics[width=0.98\textwidth]{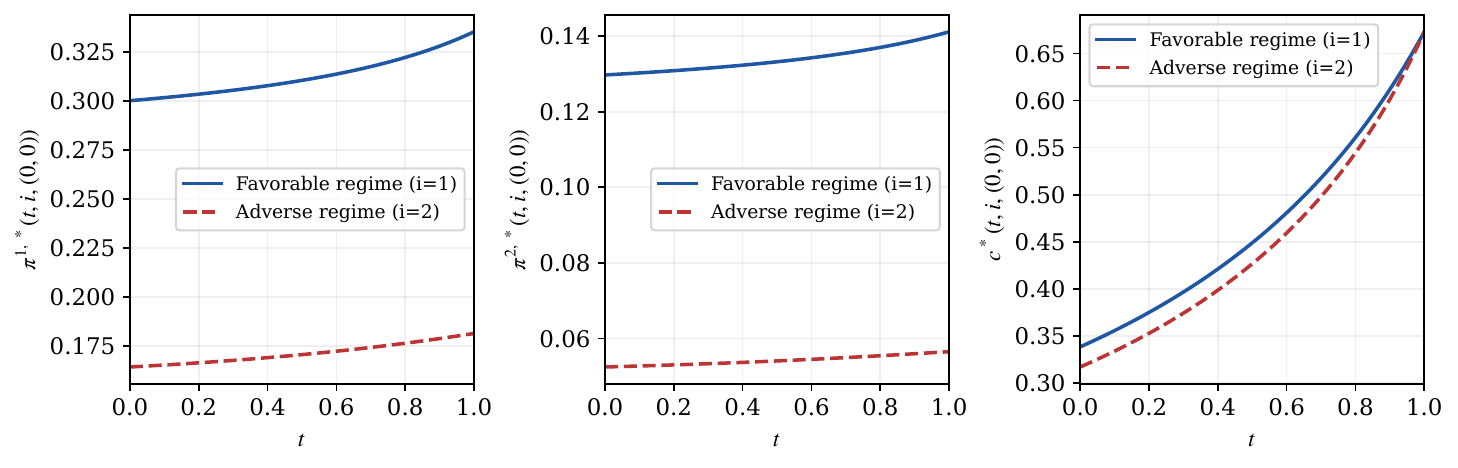}
\caption{Optimal controls in the favorable and adverse regimes with state $(0,0)$.}
\label{fig:matched-regimes}
\end{figure}

\subsection{Comparison of utility criteria}\label{subsec:matched-objectives}
To identify how intermediate consumption changes credit allocation, we compare our objective with the corresponding terminal-wealth-only problem,
\[
\sup_{\pi}\mathbb E_{t,x,i,h}\Big[\frac{(X_T^{\pi,0})^{1-\gamma}}{1-\gamma}\Big],
\]
which shares the same contagion specification, market parameters, horizon, risk-aversion parameter, and portfolio constraints but removes the running utility term entirely. Figure~\ref{fig:matched-objectives} compares the resulting optimal controls. Introducing consumption raises the portfolio fractions at $t=0.3$ in regime 1 from $(0.2972,0.1287)$ to $(0.3056,0.1315)$ and in regime 2 from $(0.1632,0.0523)$ to $(0.1676,0.0533)$. Thus, at a given wealth level, the investor combines positive consumption with larger risky-asset fractions and a smaller riskless asset fraction: introducing consumption reshapes the portfolio policy rather than merely deducting expenditure from an otherwise unchanged one.

\begin{figure}[H]
\centering
\includegraphics[width=0.98\textwidth]{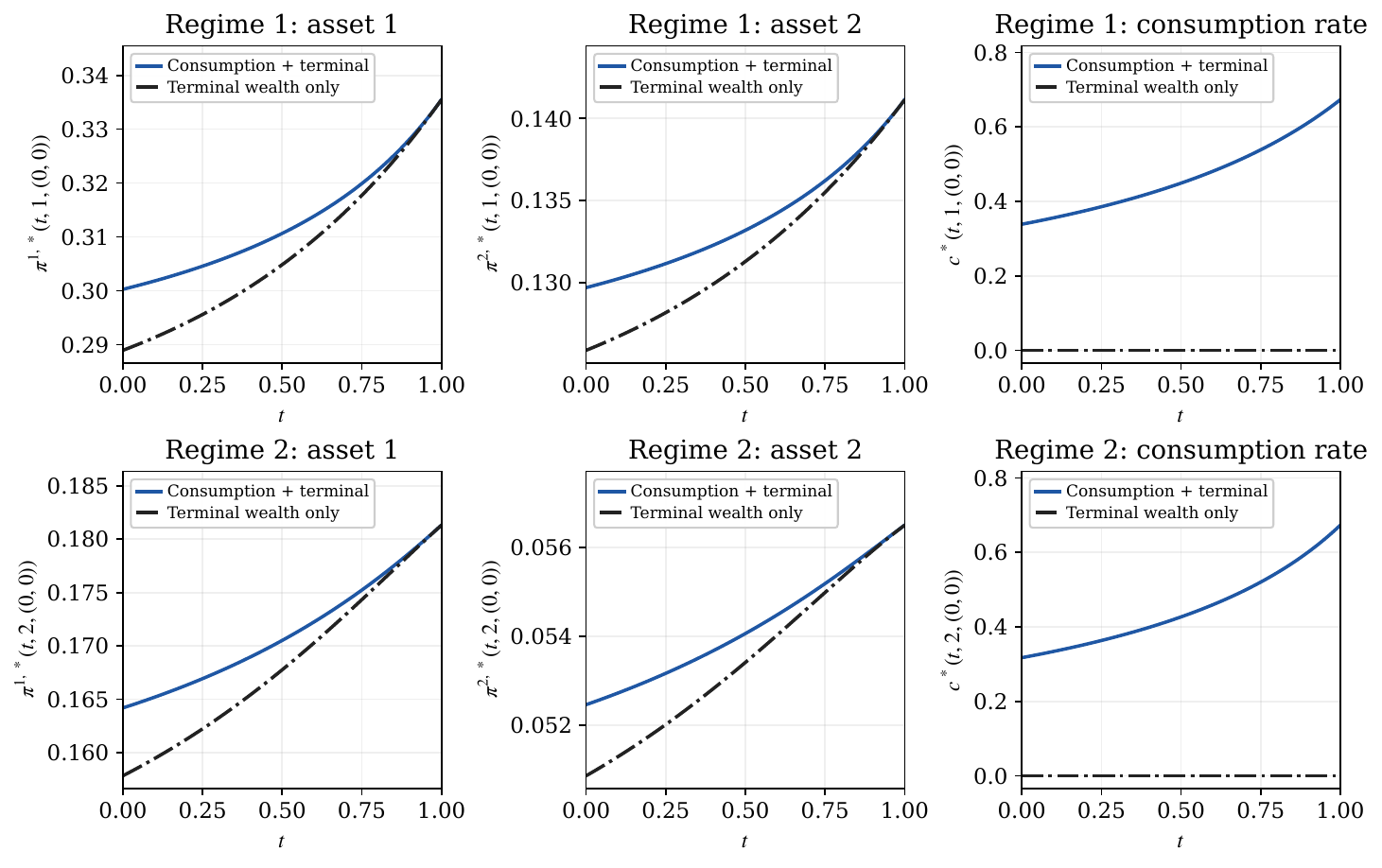}
\caption{Optimal controls under consumption plus terminal-wealth utility and terminal-wealth utility only in state $(0,0)$.}
\label{fig:matched-objectives}
\end{figure}

\subsection{Comparison before and after default}\label{subsec:matched-default}
We finally compare the optimal decisions immediately before and after a default. Figure~\ref{fig:matched-default} shows the surviving asset's optimal fraction and the consumption rate when either asset defaults at $t=0.3$, with the regime held fixed; the left and right columns correspond to defaults of assets 1 and 2, respectively. In each regime, both the surviving asset's weight and the consumption rate fall after the default. The portfolio response reflects contagion: the surviving asset's default intensity jumps upward, which discourages credit investment. At the same time, the investor consumes a smaller share of the remaining wealth. When asset 1 defaults in regime 1, for instance, the weight of asset 2 falls from $0.1315$ to $0.1248$ and the consumption rate from $0.3966$ to $0.3670$; combined with the loss of the position in asset 1, this reduces actual consumption expenditure by approximately $35.8\%$. A default therefore reshapes both the allocation and the consumption of the investor's remaining wealth.

\begin{figure}[H]
\centering
\includegraphics[width=0.98\textwidth]{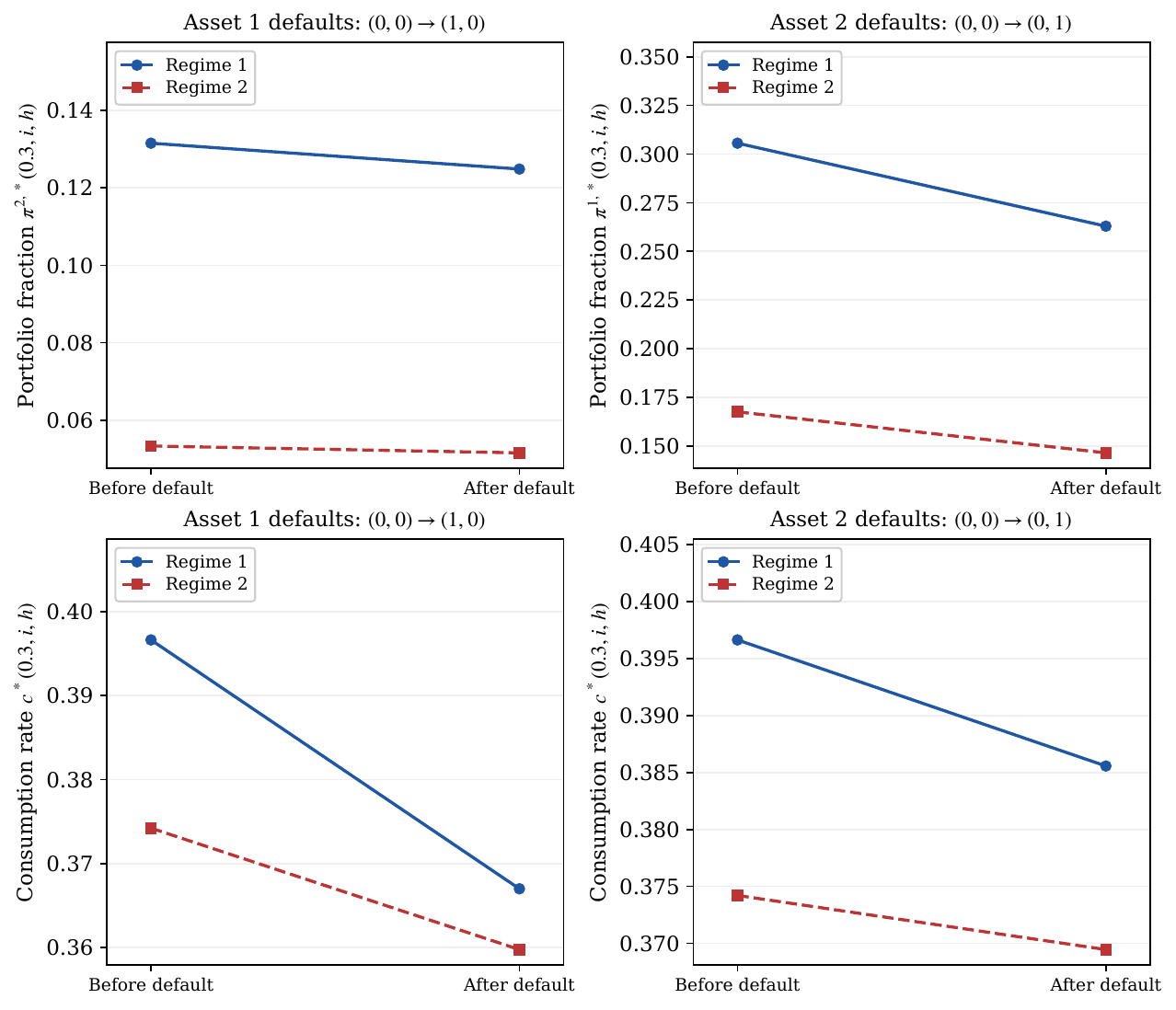}
\caption{Optimal portfolio fractions and consumption rates before and after default under contagion at $t=0.3$.}
\label{fig:matched-default}
\end{figure}

\end{document}